%% file: article-JRSS-B.tex
\documentclass[a4paper,11pt]{article}\usepackage[]{graphicx}\usepackage[]{color}
\makeatletter
\def\maxwidth{ %
  \ifdim\Gin@nat@width>\linewidth
    \linewidth
  \else
    \Gin@nat@width
  \fi
}
\makeatother

\definecolor{fgcolor}{rgb}{0.345, 0.345, 0.345}

\usepackage{framed}
\makeatletter
 {\par\unskip\endMakeFramed%
 \at@end@of@kframe}
\makeatother

\definecolor{shadecolor}{rgb}{.97, .97, .97}
\definecolor{messagecolor}{rgb}{0, 0, 0}
\definecolor{warningcolor}{rgb}{1, 0, 1}
\definecolor{errorcolor}{rgb}{1, 0, 0}
\newenvironment{knitrout}{}{} 

\usepackage{alltt}

\usepackage{amsthm,amsmath}
\usepackage[round]{natbib}
\usepackage[colorlinks,citecolor=blue,urlcolor=blue]{hyperref}

\usepackage{xr} 
\usepackage[top=1in, bottom=1in, left=1in, right=1in]{geometry}

\usepackage{import}   
\usepackage{amsfonts}
\usepackage{ae,aecompl}
\usepackage{graphicx}
\usepackage{textcomp} 
\usepackage{stmaryrd} 
\usepackage{float}
\usepackage{subfig}
\usepackage{xspace}
\usepackage{csquotes}

\usepackage[compact]{titlesec}

\let\originalparagraph\paragraph
\renewcommand{\paragraph}[1]{\originalparagraph{#1.}}

\theoremstyle{plain} 
\newtheorem{lemma}{Lemma}[section]
\newtheorem{proposition}{Proposition}[section]
\newtheorem{theorem}{Theorem}[section]
\newtheorem{corollary}{Corollary}[section]

\theoremstyle{definition} 
\newtheorem{definition}{Definition}[section]
\theoremstyle{remark} 
\newtheorem{remark}{Remark}[section]
\newtheorem{example}{Example}[section]
\numberwithin{equation}{section}
\newcommand{\field}[1]{\mathbb{#1}}
\newcommand{\R}{\field{R}}
\newcommand{\N}{\field{N}} 
\newcommand{\Lr}{\mathcal{L}} 

\newcommand{\Kr}{\mathcal{K}}
\newcommand{\Tr}{\mathcal{T}}
\newcommand{\Sr}{\mathcal{S}}
\newcommand{\Cr}{\mathcal{C}}
\newcommand{\Mr}{\mathcal{M}}
\newcommand{\Dr}{\mathcal{D}}

\newcommand{\Normal}{\mathcal{N}}
\newcommand{\Ncal}{\mathcal{N}}

\newcommand{\eqdefrev}{=\mathrel{\mathop:}}
\newcommand{\bij}{\underset{\sim}{\to}}
\newcommand{\sachant}[2]{\left.#1\mathrel{}\middle|\mathrel{}#2\right.} 

\newcommand{\Esp}{\field{E}}
\newcommand{\Espe}[1]{\field{E}\left[ #1 \right]}

\newcommand{\vect}[1]{\mathbf{#1}}
\newcommand{\vvect}[1]{\boldsymbol{#1}}
\newcommand{\1}{\vect{1}}
\newcommand{\matr}[1]{\mathbf{#1}}

\DeclareMathOperator{\Span}{Span}
\DeclareMathOperator{\Diag}{Diag}
\DeclareMathOperator{\Supp}{Supp}

\DeclareMathOperator{\Proj}{Proj}
\DeclareMathOperator*{\argmax}{argmax}
\DeclareMathOperator*{\argmin}{argmin}

\newcommand{\intervalle}[4]{\mathopen{#1}#2\mathclose{}\mathpunct{},#3\mathclose{#4}}
\newcommand{\intervalleff}[2]{\intervalle{[}{#1}{#2}{]}}

\newcommand{\intervalleentier}[2]{\intervalle\llbracket{#1}{#2}\rrbracket}
\newcommand{\floor}[1]{\mathopen{\lfloor}#1\mathclose{\rfloor}}

\newcommand{\norm}[1]{\left\lVert#1\right\rVert}
\newcommand{\mahanorm}[2]{\left\lVert#1\right\rVert_{\matr{#2}^{-1}}}
\newcommand{\card}[1]{\left\lvert#1\right\rvert}

\DeclareMathOperator{\pa}{pa}
\DeclareMathOperator{\Par}{Anc}

\DeclareMathOperator{\pen}{pen}
\DeclareMathOperator{\Crit}{Crit}
\DeclareMathOperator{\Dkhi}{Dkhi}
\DeclareMathOperator{\EDkhi}{EDkhi}

\newcommand{\Ibb}{\field{I}}
\newcommand{\OU}{Ornstein-Uhlenbeck\xspace}

\usepackage{mathtools}

\newcommand\coolrightbrace[2]{%
\left.\vphantom{\begin{matrix} #1 \end{matrix}}\right\}#2}
  
\newcommand{\PB}[1]{{#1}}
\newcommand{\NR}[1]{\textcolor{red}{#1}}

\newcommand{\printR}[1]{{\sf #1}}
\newcommand{\citeR}[1]{\printR{#1} \citep{#1}}
\newcommand{\comment}[1]{}
\IfFileExists{upquote.sty}{\usepackage{upquote}}{}
\begin{document}


\begin{center}
{\Large
  {\sc  Detection of adaptive shifts on phylogenies using shifted stochastic processes on a tree}
}
\bigskip

 Paul Bastide$^{1,2}$, Mahendra Mariadassou $^{2}$ \& St\'ephane Robin $^{1}$
\bigskip

{\it
$^{1}$ UMR MIA-Paris, AgroParisTech, INRA, Universit\'e Paris-Saclay, 75005, Paris, France



$^{2}$ MaIAGE, INRA, Universit\'e Paris-Saclay, 78352 Jouy-en-Josas, France

paul.bastide@agroparistech.fr, mahendra.mariadassou@jouy.inra.fr, stephane.robin@agroparistech.fr
}
\end{center}
\bigskip


{\bf Summary.} Comparative and evolutive ecologists are interested in the distribution of quantitative traits among related species. The classical framework for these distributions consists of a random process running along the branches of a phylogenetic tree relating the species. We consider shifts in the process parameters, which reveal fast adaptation to changes of ecological niches. We show that models with shifts are not identifiable in general. Constraining the models to be parsimonious in the number of shifts partially alleviates the problem but several evolutionary scenarios can still provide the same joint distribution for the extant species. We provide a recursive algorithm to enumerate all the equivalent scenarios and to count the number of effectively different scenarios. We introduce an incomplete-data framework and develop a maximum likelihood estimation procedure based on the EM algorithm. Finally, we propose a model selection procedure, based on the cardinal of effective scenarios, to estimate the number of shifts and for which we prove an oracle inequality.

\smallskip







\section{Introduction}
  \subsection{Motivations: Environmental Shifts}
  

  An important goal of comparative and evolutionary biology is to
  decipher the past evolutionary mechanisms that shaped present day
  diversity\PB{, and more specifically to detect the dramatic changes
    that occurred in the past \citep[see for instance][]{losos1990,
      mahler2013, davis2007, jaffe2011}}.  It is well established that
  related organisms do not evolve independently
  \citep{felsenstein1985}: their shared evolutionary history is well
  represented by a phylogenetic tree. \PB{In order to explain the
    pattern of traits measured on a set of related species, one needs
    to take these correlations into account. Indeed, a given species
    will be more likely to have a similar trait value to her
    \enquote{sister} (a closely related species) than to her
    \enquote{cousin} (a distantly related species), just because of
    the structure of the tree.  On top of that structure, when
    considering a \emph{functional} trait (i.e.\ a trait
    directly linked to the fitness of its bearer), such as shell size
    for turtles \citep{jaffe2011}, one needs to take into account the
    effect of the species environment on its traits. Indeed, a change
    in the environment for a subset of species, like a move to the
    Gal\`apagos Islands for turtles, will affect the observed trait
    distribution, here with a shift towards giant shell sizes compared
    to mainland turtles. The observed present-day trait distribution
    hence contains the footprint of adaptive events, and should allow
    us to detect unobserved past events, like the migration of one
    ancestral species to a new environment.} Our goal here is to
  devise a statistical method \PB{based on a rigorous maximum
    likelihood framework} to automatically detect the past
  environmental shifts that shaped the present day trait distribution.


    

  \subsection{Stochastic Process on a tree}
  
  We model the evolution of a quantitative adaptive trait using the
  framework of stochastic processes on a tree. Specifically, given a
  rooted phylogenetic tree, we assume that the trait evolves according
  to a given stochastic process on each branch of the tree. At each
  speciation event, or equivalently node of the tree, one independent
  copy with the same initial conditions and parameters is created for
  each daughter species, or outgoing branches.
  
  \paragraph{Tree Structure} This model is our null model: it accounts
  for the {tree-induced} {distribution} of trait values {in the
    absence of shifts}. Depending on the phenomenon studied, several
  stochastic processes can be used to capture the dynamic of the trait
  evolution. In the following, we will use the Brownian Motion (BM)
  and the \OU (OU) processes.

    \begin{figure}[!ht]
    \begin{tabular}{lr}
    \subfloat[A phylogenetic tree.]{
\begin{knitrout}
\definecolor{shadecolor}{rgb}{0.969, 0.969, 0.969}\color{fgcolor}
\includegraphics[width=0.42\textwidth]{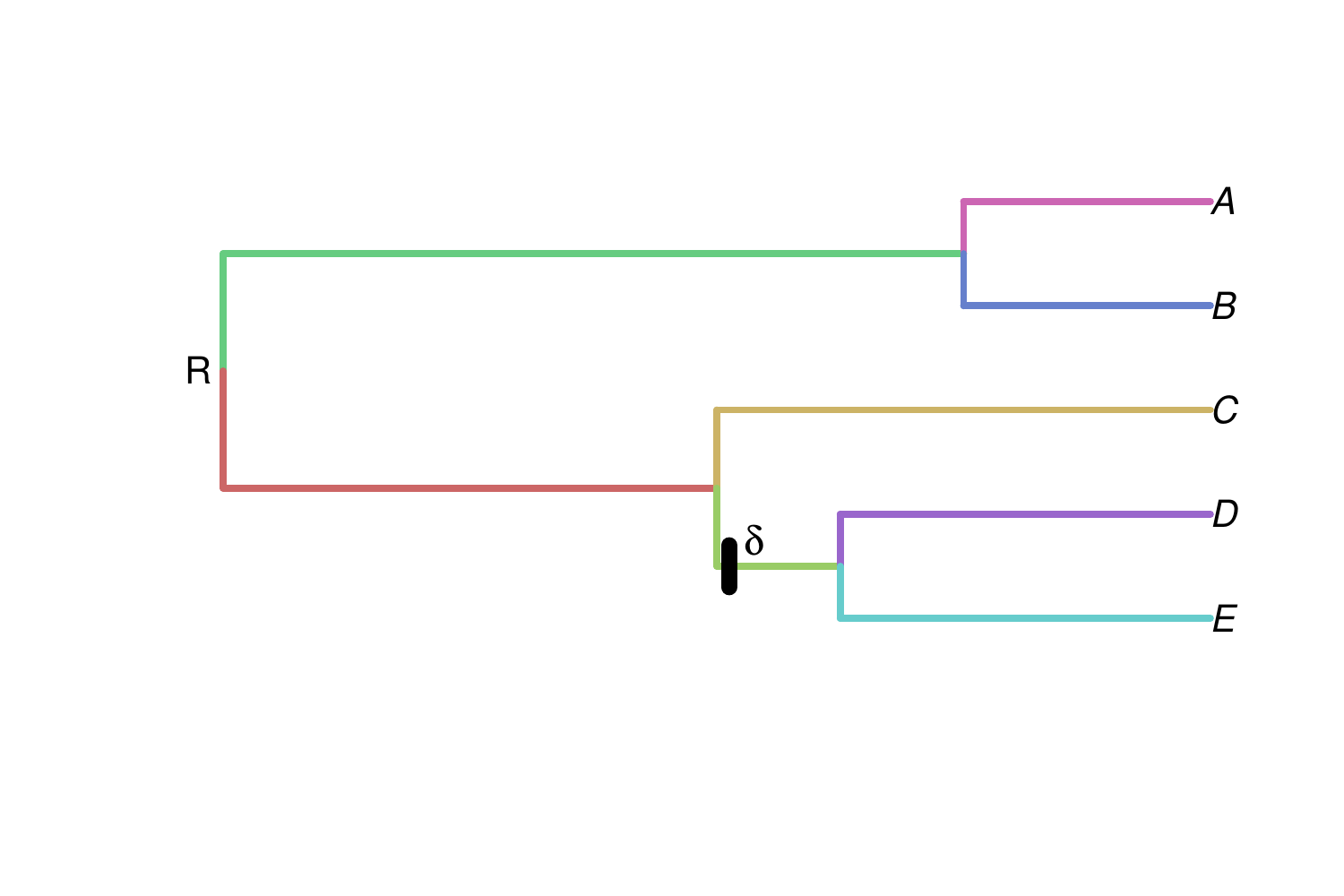} 

\end{knitrout}
} &
    \subfloat[Trait value evolution.]{
\begin{knitrout}
\definecolor{shadecolor}{rgb}{0.969, 0.969, 0.969}\color{fgcolor}
\includegraphics[width=0.42\textwidth]{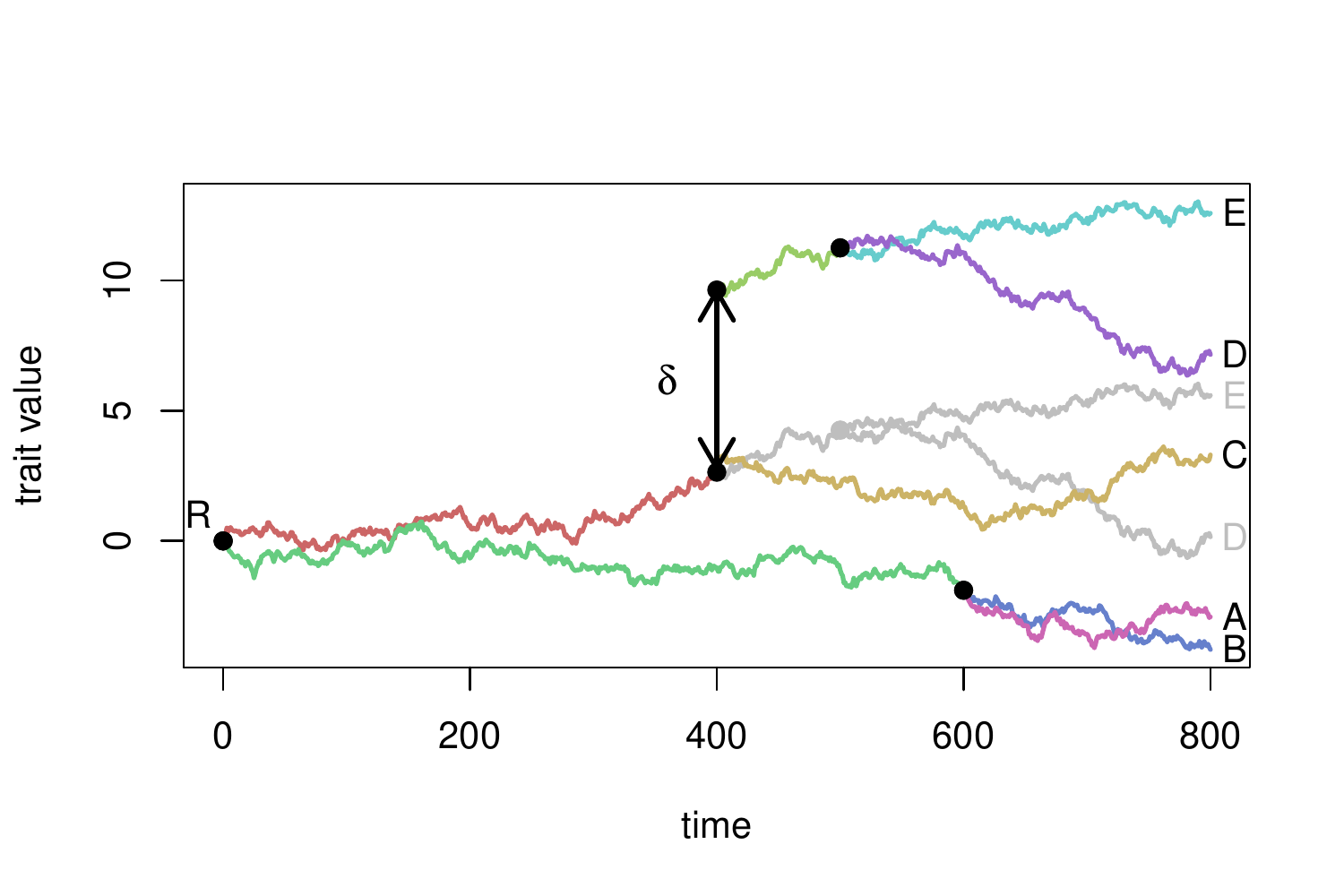} 

\end{knitrout}
    }\\
    \end{tabular}
    \caption{Trait evolution under a Brownian Motion. The ancestral value of the trait is $0$, and the observed values (time $800$) range from $-4$ to $11$ for extant species. One shift occurs on the parent branch of (D,E), changing the trajectory of their ancestral trait value from the grey one to the colored one. The shift increases the observed dispersion.}\label{fig:stochastic_process}
    \end{figure}
      
    \paragraph{Brownian Motion} Since the seminal article of \citet{felsenstein1985}, the BM has been used as a neutral model of {trait} evolution. If $(B_t; t\geq 0)$ is the Brownian motion, a character $(W_t; t\geq 0)$ evolves on a lineage according the the stochastic differential equation: $dW_t = \sigma dB_t$, $\sigma^2$ being a variance parameter. If $\mu$ is the ancestral {value} at the root of the tree ($t = 0$), {then} $W_t \sim \Normal(\mu, \sigma^2 t)$. The variance {$\sigma^2 t$} of the trait is proportional to the time of evolution and the covariance {$\sigma^2 t_{ij}$} between two species $i$ and $j$ is proportional to their time of shared evolution.
    
    \paragraph{\OU Process} An unbounded variance is quite unrealistic for adaptive traits \citep{butler2004}. For that reason, the OU process, that models stabilizing selection around an adaptive optimum \citep{hansen1997} is usually preferred to the BM. It is defined by the stochastic differential equation
    \(
      dW_t = -\alpha(W_t-\beta)dt + \sigma dB_t
    \), 
and has stationary distribution $\Normal(\beta, {\sigma^2}/{2\alpha})$. In this equation, $W_t$ is the \emph{secondary optimum} of a species, a trade-off between all selective constraints
\PB{--~e.g.\ ecological~--} 
on the trait and can be approached by the population mean of that species. The term $-\alpha(W_t-\beta)dt$ of the equation represents the effects of stabilizing selection towards a \emph{primary optimum} $\beta$, that depends only on the ecological niche of the species. The selection strength is controlled by the call-back parameter $\alpha$. For interpretation purpose, we will use the \emph{phylogenetic half-life} $t_{1/2} = \ln(2)/\alpha$, defined as the time needed for the expected trait value to move half the distance from the ancestral state to the primary optimum \citep{hansen1997}. The term $\sigma dB_t$ represents the random effects of uncontrolled factors, ranging from genetic drift to environmental fluctuations. We refer to \citet{hansen1997, hansen2008} for further discussion and deeper biological interpretations of the hypothesis underlying this model of evolution.
    The aim of our work is to detect environmental shifts.
    
    \paragraph{Environmental Shifts} In addition to the previous mechanisms, we assume that abrupt environmental changes affected the ecological niche of some ancestral species. We model these changes as instantaneous shifts in the parameters of the stochastic process. Shifted parameters are inherited along time and thus naturally create clusters of extant species that share the same parameters trajectories.
    In the BM process, shifts affect the mean value of the trait and are thus instantaneously passed on to the trait itself (see Figure~\ref{fig:stochastic_process}) whereas in the OU process, shifts affect the primary optimum $\beta$. In this case, the trait converges to its new stationary value with an exponential decay of half-life $t_{1/2}$ inducing a lag that makes recent shifts harder to detect \citep{hansen2012}. In the remainder, we assume that all other parameters ($\sigma^2$ for the BM and $\sigma^2, \alpha$ for the OU) are fixed and constant \citep[but see][for partial relaxations of this hypothesis]{beaulieu2012, rabosky2014}.
  
  \subsection{Scope of this article}
  
  \paragraph{State of the Art} Phylogenetics Comparative Methods (PCM)
  are an active field that has seen many fruitful developments in the
  last few years \citep[see][for an extensive review]{pennell2013}.
  Several methods have been specifically developed to study adaptive
  evolution, starting with the work of \citet{butler2004}.
  \citet{butler2004} only consider shifts in the optimal value
  $\beta$ whereas \citet{beaulieu2012} also allow for shifts in the
  selection strength $\alpha$ and the variance $\sigma^2$. Both have
  in common that shift locations are assumed to be known. Several
  extensions of the model without or with known shifts have also been
  proposed: \citet{hansen2008} extended the original work of
  \citet{hansen1997} on OU processes to a two-tiered model where
  $\beta(t)$ is itself a stochastic process (either BM or OU).
  \citet{bartoszek2011} extended it further to multivariate traits
  whereas \citet{hansen2012} introduced errors in the
  observations. Expanding upon the BM, \citep{landis2013} replaced
  fixed shifts, known or unknown, by random jump processes using Levy
  processes. \PB{Non-Gaussian models of trait evolution were also
    recently considered by \citet{hiscott2015}, who adapted
    Felsenstein's pruning algorithm for the likelihood computation of
    these models, using efficient integration techniques.} Finally,
  \citet{hoane2013a} derived consistency results for estimation of the
  parameters of an OU on a tree and \citet{bartoszek2012, sagitov2012}
  computed confidence intervals of the same parameters by assuming
  \PB{an unknown random tree} topology and averaging over it.

  The first steps toward automatic detection of shifts\PB{, which is
    the problem of interest in this paper,} have been done in a
  Bayesian framework, for both the BM \citep{eastman2013} and the OU
  \citep{uyeda2014}.  Using RJ-MCMC, they provide the user with the
  posterior distribution of the number and location of shifts on the
  tree.  Convergence is however severely hampered by the size of the
  search space. The growing use of PCM in fields where large trees
  are the norm makes maximum likelihood
  based point estimates of
  the shift locations more practical. A stepwise selection procedure
  for the shifts has been proposed in \citet{SURFACE}. The procedure
  adds shifts one at the time and is therefore \PB{rather} efficient
  but the selection criterion is heuristic and has no theoretical
  grounding for that problem, where observations are correlated
  through the tree structure. 
  These limitations have been pointed
    out in \citet{hoane2014}. In this article, the authors describe
    several identifiability problems that arise when trying to infer
    the position of the shifts on a tree, and propose a different
    stepwise algorithm based on a more stringent selection criterion,
    heuristically inspired by segmentation algorithms.

To rigorously tackle this issue, we introduced a framework where a univariate trait evolves according to an OU process with stationary root state (S) on an Ultrametric tree (U). Furthermore, as the exact position of a shift on a branch is not identifiable for an ultrametric tree, we assume that shifts are concomitant to speciation events and only occur at Nodes (N) of the tree. We refer to this model as OUsun hereafter.

    
    \paragraph{Our contribution} \PB{In this work, we make several major contributions to the problem at hand. First, we derive a statistical method to find a maximum likelihood estimate of the parameters of the model. When the number of shifts is fixed, we work out an Expectation Maximization (EM) algorithm that takes advantage of the tree structure of the data to efficiently maximize the likelihood. 
    Second, we show that, given the model used and the kind of data available, some evolutionary scenarios remain indistinguishable. Formally, we exhibit some identifiability problems in the location of the shifts, even when their number is fixed, and subsequently give a precise characterization of the space of models that can be inferred from the data on extant species.
    Third, we provide a rigorous model selection criterion to choose the number of shifts needed to best explain the data. Thanks to our knowledge of the structure of the spaces of models, acquired through our identifiability study, we are able to mathematically derive a penalization term, together with an oracle inequality on the estimator found.
    Fourth and finally, we implement the method on the statistical software \citeR{R}, and show that it correctly recovers the structure of the model on simulated datasets. When applied to a biological example, it gives results that are easily interpretable, and coherent with previously developed methods. All the code used in this article is publicly available on GitHub (\url{https://github.com/pbastide/PhylogeneticEM}).}
    
  
  \paragraph{Outline} In Section~\ref{sec:statistical_modeling}, we present the model, using two different mathematical point of views, that are both useful in different aspects of the inference. In Section~\ref{sec:identifiability}, we tackle the identifiability problems associated with this model, and describe efficient algorithms to enumerate, first, all equivalent models within a class, and, second, the number of truly different models for a given number of shifts. These two sections form the foundation of Section~\ref{sec:inference}, in which we describe our fully integrated maximum likelihood inference procedure. Finally, in Sections~\ref{sec:simulations} and \ref{sec:chelonia}, we conduct some numerical experiments on simulated and biological datasets.
  
\section{Statistical Modeling} \label{sec:statistical_modeling}
  
  \subsection{Probabilistic Model} \label{subsec:probabilistic_model}
  
    \begin{figure}[htb]
      \begin{flushright}
  \def\svgwidth{0.8\textwidth}
  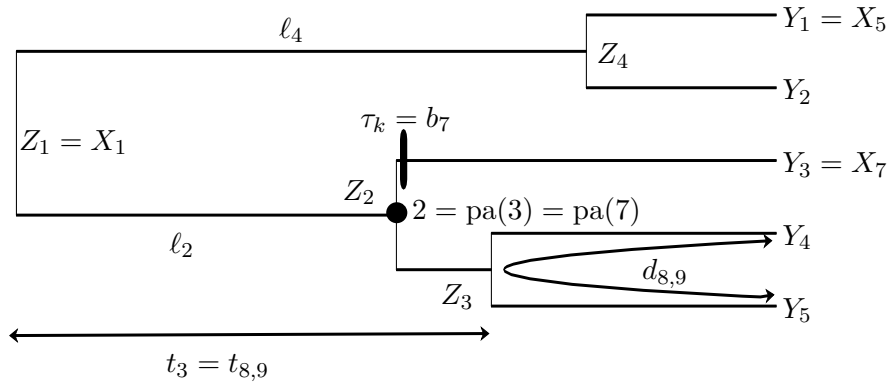
      \end{flushright}
      \caption{A rooted and time calibrated phylogenetic tree with the notations used to parametrize the tree ($l, t, d, b$) and the observed ($Y$) and non-observed ($Z$) variables.}
      \label{fig:tree_notations}
    \end{figure}
  
    \paragraph{Tree Parametrization} As shown in
    Figure~\ref{fig:tree_notations}, we consider a rooted tree $\Tr$
    with $n$ tips and $m$ internal nodes ($m = n - 1$ for binary
      trees). The internal nodes are numbered from $1$ (the root) to
    $m$, and the tips from $m + 1$ to $m + n$. Let $i$ be an integer,
    $i \in \intervalleentier{2}{m + n}$. Then $\pa(i)$ is the unique
    parent of node $i$. The branch leading to $i$ from $\pa(i)$ is noted
    $b_i$ and has length $\ell_i = t_i - t_{\pa(i)}$ where $t_i$ is
    the time elapsed between the root and node $i$. By convention, we
    set $t_1 = 0$ and $t_{\pa(1)} = - \infty$ for the root. The last
      convention ensures that the trait follows the stationary
      distribution (if any) of the process at the root. We denote
    $\Par(i) = \{\pa^r(i): r \geq 0\}$ the \PB{set} composed of node
    $i$ and of all its ancestors up to the root. For a couple of
    integers $(i, j)$, $(i, j) \in \intervalleentier{1}{m + n}^2$,
    nodes $i$ and $j$ are at phylogenetic distance $d_{ij}$ and the
    time of their most recent common ancestor (mrca) is $t_{ij}$. We
      consider ultrametric trees, for which $t_{m+1} = \dotsb =
      t_{m+n} \eqdefrev h$ and note $h$ the tree height. In the
    following, the tree is fixed and assumed to be known.\par
    
    \paragraph{Trait Values}
    We denote by $\vect{X}$ the vector of size $m + n$ of the trait values at the nodes of the tree. We split this vector between non-observed values $\vect{Z}$ (size $m$) at the internal nodes, and observed values $\vect{Y}$ (size $n$) at the tips, so that $\vect{X}^T = (\vect{Z}^T, \vect{Y}^T)$. According to our model of trait evolution, the random variable $X_i$, $i \in \intervalleentier{1}{m + n}$, is the result of a stochastic process stopped at time $t_i$.
    In the following, we assume that \PB{the inference in the BM case is done conditionally to a fixed root value $X_1 = \mu$. In the OUsun case, we assume that the root trait value is randomly drawn from the stationary distribution: $X_1 \sim \Normal(\mu = \beta_1, \gamma^2 = \frac{\sigma^2}{2\alpha})$, where $\beta_1$ is the ancestral optimal value.}
    
    
    \paragraph{Shifts}
    We assume that $K$ shifts occur on the tree, $K\in\N$. The $k^{\text{th}}$ shift, $k\in\intervalleentier{1}{K}$, occurs at the beginning of branch $\tau_k$, $\tau_k \in \{b_i, i \in \intervalleentier{2}{m+n}\}$, and has intensity $\delta_k$, $\delta_k \in \R$. The interpretation of this intensity depends on the process. In the following, we use the vector $\vect{\Delta}$ of shifts on the branches, of size $m+n$, with $K+1$ non-zero entries, and defined as follows (see example~\ref{ex:tree_matrix}):
       \[
        \Delta_1 = \mu \PB{~( = \beta_1 \text{ for an OUsun})} \quad \text{ and } \quad
		    \forall i \in \intervalleentier{2}{n+m}, \;
        \Delta_i = \begin{cases}
  	    \delta_k & \text{ if } \tau_k = b_i\\
		    0 & \text{ otherwise.}
		    \end{cases}
  \]
    Note that no proper shift occurs on the root branch, but that \PB{the root trait value or mean, $\mu$,} is formalized as \PB{an initial fictive shift on this fictive branch}.
    
    \paragraph{Parameters} The parameters needed to describe an OUsun (respectively, a BM) are
    $\vvect{\theta} = (\gamma, \alpha, \vect{\Delta})$ (resp.\ $\vvect{\theta} = (\sigma, \vect{\Delta})$). \PB{Note that, as $\sigma^2 = 2 \alpha \gamma^2$, only the two parameters $\alpha$ and $\gamma$ are needed to describe the OUsun.} We denote by OUsun$(\vvect{\theta})$ (resp.\ BM$(\vvect{\theta})$) the OUsun (resp.\ BM) process running on the tree with parameters $\vvect{\theta}$.
    
%
  
  \subsection{Incomplete Data Model Point of View}\label{subsec:latent_model}
    If the trait values were observed at all nodes of the tree, including ancestral ones, shifts would be characterized by unexpectedly large differences between a node and its parent. A way to mimic this favorable case is to use an incomplete data model, as described below. This representation of the model \PB{will be useful for the} parametric inference using an EM algorithm (Section~\ref{subsec:EM}).
    \paragraph{Brownian Motion} As the shifts occur directly in the mean of the process, we get:
    \begin{equation}\label{eq:BM_rec_def}
	\PB{X_1 = \mu}
	\quad \text{ and } \quad
	\forall i \in \intervalleentier2{m+n}, \;
	X_i | X_{\pa(i)} \sim \Normal\left(X_{\pa(i)} + \PB{\Delta_i} , ~ \ell_i \sigma^2\right)
    \end{equation}
    The trait value at node $i$, $i \in \intervalleentier2{m+n}$, is centered on the value of its parent node $X_{\pa(i)}$, with a variance proportional to the evolution time $\ell_i$ between {$i$ and $\pa(i)$}. The effect of a \PB{non-zero shift $\Delta_i$}
on branch $b_i$ is simply to translate the trait value by \PB{$\Delta_i$}.
    
    \paragraph{\OU} The shifts occur on the primary optimum $\vvect{\beta}$, which is piecewise constant. As the shifts are assumed to occur at nodes, the primary optimum is entirely defined by its initial value $\beta_1$ \PB{and its values $\beta_2, \dots, \beta_{n+m}$ on branches of the tree, where $\beta_i$ is the value on branch $b_i$ leading to node $i$.}
    \begin{equation}\label{eq:OU_beta_rec_def}
	\NR{\beta_1 \in \R ~( = \mu \text{ for an OUsun})}
	\quad \text{ and } \quad
	\forall i \in \intervalleentier2{m+n}, \;
	\beta_i = \beta_{\pa(i)} + \PB{\Delta_i}
    \end{equation}
    Assuming that the root node is in the stationary state, we get:
    \begin{equation}\label{eq:OU_rec_def}
      \begin{cases}
	X_1 \sim \Ncal(\mu = \beta_1, \gamma^2 = \frac{\sigma^2}{2\alpha}) & \\
	X_i | X_{\pa(i)} \sim \Normal\left(X_{\pa(i)}e^{-\alpha \ell_i} + \beta_{i}(1 - e^{-\alpha \ell_i})
	, ~ \frac{\sigma^2}{2\alpha} (1 - e^{-2\alpha \ell_i})\right) & \forall i \in \intervalleentier2{m+n}
      \end{cases}
    \end{equation}
    The trait value \PB{at} node $i$ depends on both the trait value at the father node $X_{\pa(i)}$ and the value $\beta_i$ of the primary optimum on branch $b_i$. Contrary to the BM case, the shifts only appear indirectly {in the distributions of $X_i$s}, through the values of $\vvect{\beta}$, and with a shrinkage of $1-e^{-\alpha d}$ for shifts of age $d$, which makes recent shifts ($d$ small compared to $1/\alpha$) harder to detect. 
  
  \subsection{Linear Regression Model Point of View}\label{subsec:linear_model}
  A more compact and direct representation of the model is to use the tree incidence matrix to link linearly the observed values (at the tips) with the shift values, as explained below. We will use this linear regression framework for the Lasso \citep{tibshirani96} initialization of the EM (Section~\ref{subsec:EM}) and the model selection procedure (Section~\ref{sec:model_selection}). It will also help us to explore identifiability issues raised in the next section.
  
    \paragraph{Matrix of a tree} It follows from the recursive definition of $\vect{X}$ that it is a Gaussian vector. In order to express its mean vector given the shifts, we introduce the tree squared matrix $\matr{U}$, of size $(m + n)$, defined by its general term: $U_{ij} = \Ibb\{j \in \Par(i)\}, \forall (i, j) \in \intervalleentier{1}{m+n}^2$. In other words, the $j^{\text{th}}$ column of this matrix, $j \in \intervalleentier{1}{m+n}$, is the indicator vector of the descendants of node $j$. To express the mean vector of the observed values $\vect{Y}$, we also need the sub-matrix $\matr{T}$, of size $n\times (m+n)$, composed of the bottom $n$ \PB{rows} of matrix $\matr{U}$, corresponding to the tips 
    (see example~\ref{ex:tree_matrix} below). Likewise, the $i^{\text{th}}$ \PB{row} of $\matr{T}$, $i\in\intervalleentier1n$, is the indicator vector of the ancestors of leaf $m + i$.\par
    
    \paragraph{Brownian Motion} From the tree structure, we get:
    \begin{equation}\label{eq:BM_lin_def_X}
     \vect{X} = \matr{U}\vect{\Delta} + \vect{E_X} 
     \qquad \text{and} \qquad
     \vect{Y} = \matr{T}\vect{\Delta} + \vect{E_Y} 
    \end{equation}
    Here, $\vect{E_X} \sim \Normal\left(\vect{0}, \matr{\Sigma_{XX}}\right)$ is a Gaussian error vector with co-variances 
    \(
    \left[\Sigma_{XX}\right]_{ij} = \PB{\sigma^2t_{ij}}
    \) 
    for any $1\leq i, j \leq m + n$, and $\vect{E_Y}$ is the vector made of the last $n$ coordinates of $\vect{E_X}$.
    
    \paragraph{\OU} For the OUsun, shifts occur on the primary optimum, and there is a lag term, so that:
    \begin{equation}\label{eq:OU_lin_def_X}
     \vvect{\beta} = \matr{U}\vect{\Delta}
     \qquad \text{and} \qquad
     \vect{X} = (\matr{U} - \matr{A}\matr{U}\matr{B})\vect{\Delta} + \vect{E_X}
    \end{equation}
    where $\matr{A}=\Diag(e^{-\alpha t_i}, 1\leq i\leq m+n)$ and $\matr{B} = \Diag(0, e^{\alpha t_{\pa(i)}}, 2\leq i\leq m+n)$ are diagonal matrices of size $m+n$. As previously, $\vect{E_X} \sim \Normal\left(\vect{0}, \matr{\Sigma_{XX}}\right)$, but
    \(
    \matr{\Sigma_{XX}} = \gamma^2[e^{-\alpha d_{ij}}]_{1\leq i, j \leq m + n}
    \). As the tree is ultrametric, this expression simplifies to the following one when considering only observed values:
    \begin{equation}\label{eq:OU_lin_def}
     \vect{Y} = \matr{T}\matr{W}(\alpha)\vect{\Delta} + \vect{E_Y}
    \end{equation}
    where $\vect{E_Y}$ is the Gaussian vector made of the last $n$ coordinates of $\vect{E_X}$, and $\matr{W}(\alpha) = \Diag(1, 1 - e^{-\alpha (h - t_{\pa(i)})}, 2\leq i\leq m+n)$ is a diagonal matrix of size $m + n$. Note that if $\alpha$ is positive, then $\alpha (h-t_{\pa(i)})>0$ for any $i \in \intervalleentier{1}{m+n}$, and $\matr{W}(\alpha)$ is invertible. 
    
        \begin{example}\label{ex:tree_matrix}
    The tree presented in Figure~\ref{fig:tree_notations} has five tips and one shift on branch $4+3=7$, 
    so:
    \[
    \matr{U}=
    \bordermatrix{
    \,\,&Z_1\!\!\!&Z_2\!\!\!&Z_3\!\!\!&Z_4\!\!\!&Y_1\!\!\!&Y_2\!\!\!&Y_3\!\!\!&Y_4\!\!\!&Y_5\cr
    Z_1\,\,&1&0&0&0&0&0&0&0&0\cr
    Z_2\,\,&1&1&0&0&0&0&0&0&0\cr
    Z_3\,\,&1&1&1&0&0&0&0&0&0\cr
    Z_4\,\,&1&0&0&1&0&0&0&0&0\cr
    \noalign{\smallskip\hrule\smallskip}
    Y_1\,\,&1&0&0&1&1&0&0&0&0 \cr
    Y_2\,\,&1&0&0&1&0&1&0&0&0 \cr
    Y_3\,\,&1&1&0&0&0&0&1&0&0 \cr
    Y_4\,\,&1&1&1&0&0&0&0&1&0 \cr
    Y_5\,\,&1&1&1&0&0&0&0&0&1 \cr 
    }
    \begin{matrix}
    \vphantom{\coolrightbrace{0\\0\\0\\0}{\matr{T}}}\\
    \coolrightbrace{0\\0\\0\\0\\1 }{\matr{T}}
    \end{matrix}
    \text{ and }
          \vect{\Delta} = \begin{pmatrix}
      \mu\\
      0\\
      0\\
      0\\
      0\\
      0\\
      \delta_1\\
      0\\
      0\\
      \end{pmatrix} 
    \]
    And, respectively, for a BM or an OUsun ($\mu = \beta_1$):
    \begin{align*}
      \Espe{\vect{Y}} & = \matr{T}\vect{\Delta} = (
      \mu, 
      \mu, 
      \mu + \delta_1, 
      \mu, 
      \mu)^T & \text{(BM)}
      \\
      \Espe{\vect{Y}} &=  \matr{T}\matr{W}(\alpha)\vect{\Delta} = (
      \beta_1,
      \beta_1,
      \beta_1 + \delta_1(1 - e^{-\alpha (h - t_{2})}),
      \beta_1,
      \beta_1)^T & \text{(OU)}
    \end{align*}
    \end{example}
    
    
    \paragraph{Space of Expectations}
    Expressions~\ref{eq:BM_lin_def_X} and~\ref{eq:OU_lin_def} allow us to link the parameter $\vvect{\theta}$ to the probability distribution of observations $\vvect{Y}$ and to explore identifiability issues. In this linear formulation, detecting shifts boils down to identifying the non-zero components of $\vect{\Delta}$. The following lemma highlights the parallels between solutions of the BM and OUsun processes:
            
      \begin{lemma}[Similar Solutions]\label{lemma:eq_BM_OUsun}
      Let $\vect{m_Y} \in \R^n$ be a vector, $\Tr$ an ultrametric tree, $\alpha$ a positive real number, and $\sigma$, $\gamma$ non-negative real numbers.
      Then there exists at least one vector $\vect{\Delta}^{BM}$, $\vect{\Delta}^{BM} \in \R^{m+n}$ (respectively, $\vect{\Delta}^{OU} \in \R^{m+n}$), such that the vector of expectations at the tips of a BM$(\sigma, \vect{\Delta}^{BM})$ (respectively, an OUsun$(\gamma, \alpha, \vect{\Delta}^{OU})$) running on the tree $\Tr$ is exactly $\vect{m_Y}$.\par
      Furthermore, $\vect{\Delta}^{BM}$ is a solution to this problem for the BM if and only if $\vect{\Delta}^{OU} = \matr{W}(\alpha)^{-1} \vect{\Delta}^{BM}$ is a solution for the OUsun, and $\vect{\Delta}^{BM}$ and $\matr{W}(\alpha)^{-1} \vect{\Delta}^{BM}$ have the same support. These two vectors are said to be \emph{similar}.
      \end{lemma}
      
      \begin{proof}
        The first part of this lemma follows directly from formulas~\ref{eq:BM_lin_def_X} (BM) and \ref{eq:OU_lin_def} (OU). Indeed, the maps $\vect{\Delta} \mapsto \matr{T}\vect{\Delta}$ and $\vect{\Delta} \mapsto \matr{T}\matr{W}(\alpha)\vect{\Delta}$ both span $\R^{n}$. The second part of the lemma is a consequence of $\matr{W}(\alpha)$ being diagonal and invertible (for $\alpha > 0$).
      \end{proof}
      
      \begin{remark}
      Lemma~\ref{lemma:eq_BM_OUsun} shows that the OUsun and BM processes that induce a given {$\vect{m_Y}$} use shifts located on the same branches, although they may differ on other parameters.
      \end{remark}
      
\section{Identifiability and Complexity of a Model} \label{sec:identifiability}

  \subsection{Identifiability Issues}\label{subsec:identifiability}
  As we only have access to $\vect{Y}$, and not $\vect{X}$, we only have partial information about the shifts occurrence on the tree. In fact, several different allocations of the shifts can produce the same trait distribution at the tips, and hence are not identifiable. \PB{In other words, there exists parameters $\vvect{\theta} \neq \vvect{\theta}'$ with the same likelihood function: $p_{\vvect{\theta}}(\cdot) = p_{\vvect{\theta}'}(\cdot)$. Note that the notion of identifiability is intrinsic to the model and affects all estimation methods.}
  \PB{Restricting ourselves to the parsimonious allocations of shifts only partially alleviates this issue, and, using a \enquote{random cluster model} representation of the problem, we are able to enumerate, first, all the equivalent solutions to a given problem, and, second, all the equivalence classes for a given number of shifts.} 
  
  \begin{figure}[htb]
	\centering
	\def\svgwidth{0.8\textwidth}
	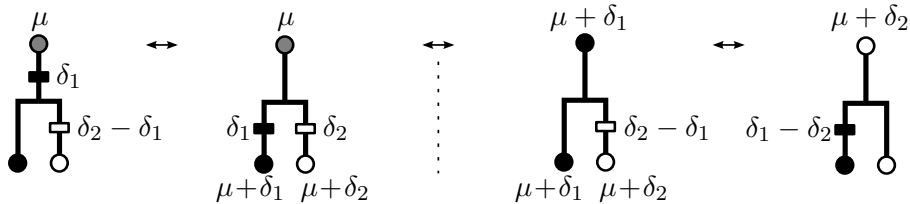
	\caption{Equivalent allocations in the BM case. Mean tip values are represented by colors and equal for all allocations. The two allocations on the right are parsimonious.}
	\label{fig:basic_equivalencies}
	\end{figure}
	
      \paragraph{No Homoplasy Assumption} We assume in the following that there is no convergent evolution. This means that each shift creates a new (and unique) mean trait value for extant species that are below it. This assumption is reasonable considering that shifts are real valued and makes the model similar to \enquote{infinite alleles} models in population genetics. This assumption confines but does not eliminate the identifiability issue, as seen in Figure~\ref{fig:basic_equivalencies}.

  \subsubsection{Definition of the problem}
  
    
Figure~\ref{fig:basic_equivalencies} shows a simple example where the model is not identifiable in the BM case. Here, four distinct allocations give the same mean values $(\mu + \delta_1, \mu + \delta_2)$ at the tips. The lack of identifiability is due to the non-invertibility of the tree matrix $\matr{T}$.
    
      \begin{proposition}[Kernel of the Tree Matrix $\matr{T}$] \label{prop:kernel_of_T}
      Let $i$ be an internal node, $i\in\intervalleentier1{m}$, with $L_i$ \PB{children} nodes $(d_1, \cdots, d_{L_i}) \in \intervalleentier2{m+n}^{L_i}$. Then the vector $\vect{K}^i$ defined as follow:
      \[
      \forall j \in \intervalleentier1{m+n}, K^i_j = \begin{cases}
      									1 & \text{if } j = i\\
									-1 & \text{if } j \in (d_1, \cdots, d_{L_i})\\
									0 & \text{otherwise}
      	       							       \end{cases}
      \]
      is in the kernel of $\matr{T}$. In addition, the $m$ vectors constructed this way form a basis of the kernel space of $\matr{T}$.
      \end{proposition}
      These kernel vectors effectively \enquote{cancel out} a shift on a branch by balancing it with the opposite shift on all immediate \PB{child} branches. Note that the root mean value is treated as a shift. 
The following lemma describes the relationships that exist between these kernel vectors and the tree matrix $\matr{U}$ defined in Section~\ref{subsec:linear_model}. 
  \NR{The technical proofs of Proposition~\ref{prop:kernel_of_T} and Lemma~\ref{lemma:U_invertible} are postponed to Appendix~\ref{supp:proofs_identifiability}.}
    
    \begin{lemma}\label{lemma:U_invertible}
      Let $b$ \PB{be} the canonical basis of $\R^{m + n}$, and $S$ a supplementary space of $\ker(\matr{T})$. Then $b' = (\vect{K}^1, \cdots, \vect{K}^m, \vect{b}_{m + 1}, \cdots, \vect{b}_{m+n})$ is a basis adapted to the decomposition $\ker(\matr{T})\oplus S$, and the matrix $\matr{U}$ (as defined in Section~\ref{subsec:linear_model}) is the \PB{change of basis} matrix between $b$ and $b'$.\\
      As a consequence, $\matr{U}$ is invertible.\\
      \end{lemma}
      
  
    \paragraph{\enquote{Random Cluster Model} Representation} When inferring the shifts, we have to keep in mind this problem of non-identifiability, and be able to choose, if necessary, one or several possible allocations among all the equivalent ones.
    In order to study the properties of the allocations, we use a \emph{random cluster model}, as defined in \citet{mossel2004}. The following definition states the problem as a node coloring problem.
    

      
      
      
      \begin{definition}[Node Coloring]
      Let $\Cr_K$ be a set of $K$ arbitrary \enquote{colors}, $K \in \N^*$. For a given shift allocation, the color of each node is given by the application $B: \intervalleentier{1}{m+n} \to \Cr_K$ recursively defined in the following way:
      \begin{itemize}
       \item Choose a color $c\in\Cr_K$ for the root: $B(1) = c$.
       \item For a node $i$, $i\in\intervalleentier2{m+n}$, set $B(i)$ to $B(\pa(i))$ if there is no shift on branch $i$, otherwise choose another color $c$, $c \in \Cr_K\setminus\{B(\pa(i))\}$, and set $B(i)$ to $c$.
      \end{itemize}
      Hereafter, we identify $(\Cr_K)^{\intervalleentier{1}{m+n}}$ with $(\Cr_K)^{m+n}$ and refer to a node coloring indifferently as an application or a vector.
      \end{definition}
      
       As the shifts only affect $\Espe{\vect{X}}$ and we only have access to $\Espe{\vect{Y}}$, we identify colors with the distinct values of $\Espe{\vect{Y}}$:
      
       \begin{definition}[Adapted Node Coloring]
         A node coloring is said to be \emph{adapted} to a shifted random process on a tree if two \emph{tips} have the same color if and only if they have the same mean value under that process.
       \end{definition}
       
       
       \begin{proposition}[Adapted Coloring for BM and OUsun]\label{prop:eq_color_shifts}
         Let $\sigma$ and $\gamma$ be two non-negative real numbers, and $\alpha$ a positive real number. Then:
         \begin{itemize}
         \item[(i)] In the BM case, if $\Cr$ is the \PB{set} of possible mean values taken by the nodes of the tree, then {the} knowledge of the {node colors} is equivalent to the knowledge of {$\vect{\Delta}$}. Furthermore, the associated node coloring is adapted to the original BM.
         \item[(ii)] In the OUsun case, from lemma~\ref{lemma:eq_BM_OUsun}, we can find a similar BM process, i.e.\ with shifts on the same branches. Then the knowledge of the node coloring associated to this similar BM process is equivalent to the knowledge of the vector of shifts of the OUsun, and the node coloring obtained is adapted to the original OUsun.
         \end{itemize}
      \end{proposition}
      
      \begin{proof}[Proof of Proposition~\ref{prop:eq_color_shifts}]
       The proof of $(i)$ relies on expression~\ref{eq:BM_lin_def_X}, that states that $\Espe{\vect{X}} = \matr{U}\vect{\Delta}$. Defining $\Cr$ as the set off all distinct values of $\Espe{\vect{X}}$, we can identify $\Espe{\vect{X}}$ with the node coloring application that maps any node $i$ with $\Espe{X}_i$. Since $\matr{U}$ is invertible (see lemma~\ref{lemma:U_invertible} above), we can go from one formalism to the other.\\
       For $(ii)$, we use lemma~\ref{lemma:eq_BM_OUsun} to find a similar BM, and then use (i).
      \end{proof}
      
\noindent From now on, we will study the problem of shifts allocation as a discrete-state coloring problem.
      
      
      
      \subsubsection{Parsimony} As we saw on Figure~\ref{fig:basic_equivalencies} there are multiple {colorings} of the internal nodes that lead to a given tips {coloring}. Among all these solutions, we choose to study only the \emph{parsimonious} ones. This property can be seen as an optimality condition, as defined below: 
      
      \begin{definition}[Parsimonious Allocation]
        Given a vector of mean values at the tips produced by a given shifted stochastic process running on the tree, an adapted node coloring is said to be \emph{parsimonious} if it has a minimum number of color changes. We denote by $\Sr_K^P$ the \PB{set} of parsimonious allocations of $K$ shifts on the $(m+n-1)$ branches of the tree (not counting the root branch).
      \end{definition}

      As $K$ shifts cannot produce more than $K+1$ colors, we can define an application $\phi: \Sr_K^P \to (\Cr_{K+1})^n$ that {maps a} parsimonious allocation of shifts {to its} associated tip partition.
      
      \begin{definition}[Equivalence]\label{def:eq_classes}
       Two allocations are said to be \emph{equivalent} (noted $\sim$) if they produce the same \PB{partition} of the tips and are both parsimonious. Mathematically:
    \[
    \forall s_1, s_2 \in \Sr_K^P, ~ s_1 \sim s_2 \iff \phi(s_1)=\phi(s_2)
    \]
\PB{In other words, two allocations are equivalent if they produce the same tip {\sl coloring} up to a permutation of the colors.}
       Given $d\in(\Cr_{K+1})^n$ a coloring of the tips of $\Tr$ with $K+1$ colors, $\phi^{-1}(d)$ is the \PB{set} of equivalent parsimonious node coloring {that coincide with $d$ \PB{(up to a permutation of the colors)} on the tree leaves}.
      \end{definition}

      Several dynamic programming algorithms already exist to compute the minimal number of shifts required to produce a given tips coloring, and to find one associated parsimonious solution \citep[see][]{fitch1971, sankoff1975, felsenstein2004}. Here, we need to be a little more precise, as we want to {both} count and enumerate all possible equivalent node coloring{s} associated with a tip coloring. For the sake of brevity, we only present the algorithm that count{s} $\card{\phi^{-1}(d)}$, for $d\in(\Cr_K)^n$. This algorithm can be seen as a corollary {of} the enumeration algorithm (presented and proved in Appendix~\ref{supp:proofs_enumaration}) and an extension of Fitch algorithm where we keep track of both the cost of an optimal coloring and the number of such colorings. It has $O(K^2Ln)$ time complexity where $L$ is the maximal \PB{number of children} of the nodes of the tree.

      \begin{proposition}[Size of an equivalence class]\label{prop:size_eq_class}
      Let $d$ be a coloring of the tips, $d\in(\Cr_K)^n$, and let $i$ be a node of tree $\Tr$ with $L_i$ daughter nodes $(i_1, \cdots, i_{L_i})$, $L_i\geq 2$. Denote by $\Tr_i$ the sub-tree rooted at node $i$.\\
      For $k\in\Cr_K$, $S_i(k)$ is the \emph{cost} of starting from node $i$ with color $k$, i.e.\ the minimal number of shifts needed to get the coloring of the tips of $\Tr_i$ defined by $d$, when starting with node $i$ in color $k$.
      Denote by $T_i(k)$ the number {of} allocations on $\Tr_i$  that {achieve cost} $S_i(k)$.\par
      If $i$ is a tip ($m+1\leq i \leq m+n$), then,\\
       \begin{minipage}{0.5\textwidth}
 \[
  S_{i}(k) = \begin{cases}
           0 & \text{ if $d(i) = k$}\\
           +\infty & \text{ otherwise}
           \end{cases}
  \]
   \end{minipage}
   \begin{minipage}{0.5\textwidth}
 \[
  T_{i}(k) = \begin{cases}
           1 & \text{ if $d(i) = k$}\\
           0 & \text{ otherwise}
           \end{cases}
  \]
   \end{minipage}
   Otherwise, if $i$ is a node, for $1\leq l \leq L_i$, define the set of admissible colors for daughter $i_l$: 
 \[
 \Kr_k^l = \argmin_{p\in\Cr_K} \left\{ S_{i_l}(p) + \Ibb\{p \neq k\} \right\} 
 \]
 As these sets are not empty, let $(p_1, \dotsc p_L) \in \Kr_k^1 \times \dotsc \times \Kr_k^L$. Then:
 \[
  S_i(k)  = \sum_{l = 1}^L S_{i_l}(p_l) + \Ibb\{p_l \neq k\}
  \quad\text{and}\quad
	T_i(k)  = \prod_{l = 1}^L \sum_{p_l \in \Kr_k^l} T_{i_l}(p_l)
	\]
  At the root, if $\Lr = \argmin_{k\in\Cr_K}S_1(k)$, then
  \(
  \card{\phi^{-1}(d)} = \sum_{k \in \Lr}T_1(k)
  \).
      \end{proposition}
      
    \paragraph{OU Practical Case} We can illustrate this notion on a simple example. We consider an OUsun on a random tree of unit \PB{height (total height $h=1$)}.
 We put three shifts on the tree, producing a given trait distribution.
 Then, using proposition~\ref{prop:eq_color_shifts} and our enumeration algorithm, we can reconstruct the $5$ possible allocations of shifts that produce the exact same distribution at the tips. These solutions are shown in Figure~\ref{fig:equivalent_shifts_ex_1}. Note that the colors are not defined by the values of the optimal regime $\vect{\boldsymbol{\beta}}$, but by the mean values $\Espe{\vect{Y}}$ of the process at the tips. As a result, the groups shown in blue and red in the first solution have the same optimal value in this configuration, but not in any other. The second solution shown illustrates the fact that all the shifts values are inter-dependent, as changing the position of only one of them can have repercussions on all the others. Finally, the third solution shows that the timing of shifts matters: to have the same impact as an old shift, a recent one must have a much higher intensity (under constant selection strength such as in the OUsun). 


\begin{figure}[!ht]

\includegraphics[width=\linewidth]{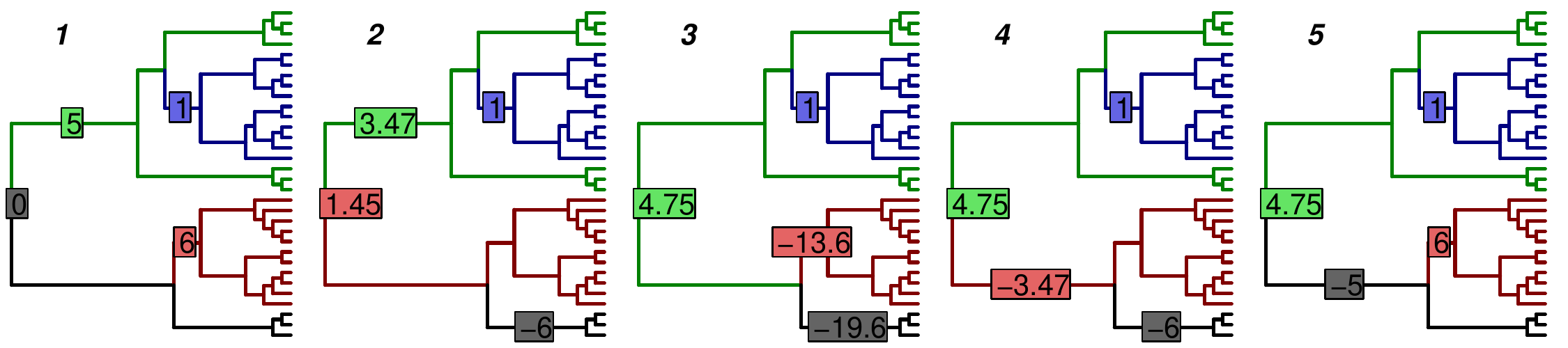} 

\caption{Five equivalent shift allocations that produce colorings that are adapted to an OUsun, with $\alpha = 3$ and $\gamma^2 = 0.1$. The box at the root represents the ancestral optimum $\beta_1$, and the boxes on the branches represent the positions and values of the shifts \PB{on the optimal value}.
While accounting for very different evolutionary scenarios, all allocations produce the same trait distribution at the tips.}\label{fig:equivalent_shifts_ex_1}
\end{figure}

\paragraph{Possible relaxation of the No-homoplasy assumption}
Note that \PB{the algorithms used for counting and enumerating the configurations of an equivalence class are valid even without} the no-homoplasy hypothesis. The no-homoplasy hypothesis is however crucial in the next Section \PB{to} establish a link between the number of shifts and the number of distinct tips colors.
  
  \subsection{Complexity of a Collection of Models}\label{subsec:complexity_model}
  
    \paragraph{Number of different tips colors} As we make the inference \PB{on} the parameters with a {fixed} number of shifts $K$  (see Section~\ref{subsec:EM}), {we need a model selection procedure to choose $K$. This procedure depends on}  the \emph{complexity} of the collection of models that use $K$ shifts, defined as the number of \emph{distinct} models. To do that, we count the number of \emph{tree-compatible} {partitions} of the tips {into} $K+1$ {groups}, as defined in the next proposition:
    \begin{proposition}\label{prop:Kshifts_K+1clusters}
    Under the no homoplasy {assumption}, an allocation of $K$ shifts on a tree is parsimonious if and only if it creates exactly $K+1$ {tip {colors}}. \PB{The tip partition into $K+1$ groups associated with this coloring is said to be \emph{tree-compatible}. The set $\Dr_{K+1} \subset (\Cr_{K+1})^n$ of such  {partitions} {is} the image of $\Sr_K^P$ by the {map} $\phi$ defined in the previous section}.
    \end{proposition}
    
        \begin{proof}[Proof of Proposition~\ref{prop:Kshifts_K+1clusters}]
    First, note that  $K$ shifts create {at most} $K+1$ {colors}. {If each shift} produce{s} a new {tip mean value} (no homoplasy), the only way to create {$K$ or less} {colors} is to \enquote{forget} one of the shifts, i.e.\ to put shifts on every descendant of the branch where it happens. Such an allocation is not parsimonious, as we could just add the value of the forgotten shift to all its descendant to get the same {coloring} of the tips with one less shift. So a parsimonious allocation cannot create less than $K+1$ {colors}, and hence creates exactly $K+1$ {colors}.\par
    Reciprocally, if an allocation with $K$ shifts that produces $p$ groups is not parsimonious, then we can find another parsimonious one that produces the same $p$ groups with $p-1$ shifts, with $p-1 < K$, i.e.\ $p < K+1$. So, by contraposition, if the allocation produces $K+1$ groups, then it is parsimonious.
    \end{proof}
      
    Using the equivalence relation defined in Definition~\ref{def:eq_classes}, we can formally take the quotient set of $\Sr_K^P$ by the relation $\sim$ to get the set of parsimonious allocations of $K$ shifts on the $m+n-1$ branches of the tree that are identifiable: \(\Sr_K^{PI} = \Sr_K^P / \sim \). In other words, 
\PB{the set $\Sr_K^{PI}$ is constituted of one representative of each equivalence class.}
Under the no homoplasy {assumption}, there is thus a bijection between identifiable parsimonious allocations of $K$ shifts and tree-compatible \PB{partitions} of the tips in $K+1$ groups: $\Sr_K^{PI} \bij \Dr_{K+1}$.\par
    The number $N^{(\Tr)}_{K+1} = \card{\Dr_{K+1}}$ is the complexity of the class of models with $K$ shifts {defined as} the number of distinct identifiable parsimonious possible configurations one can get with $K$ shifts on the tree. To compute $N^{(\Tr)}_K$, we will need $M^{(\Tr)}_K$ the number of \emph{marked} tree-compatible partitions in $K$ groups. \NR{These are composed of all the tree-compatible partitions where one group, among those that could be in the same state as the root, is distinguished with a mark} (see example~\ref{ex:marqued_unmarqued} below).
    
\begin{example}[Difference between $N^{(\Tr)}_2$ and $M^{(\Tr)}_2$]\label{ex:marqued_unmarqued}\qquad

        \begin{minipage}{0.21\textwidth}
  \begin{figure}[H]
    \includegraphics[width=\textwidth, height = 3.2cm]{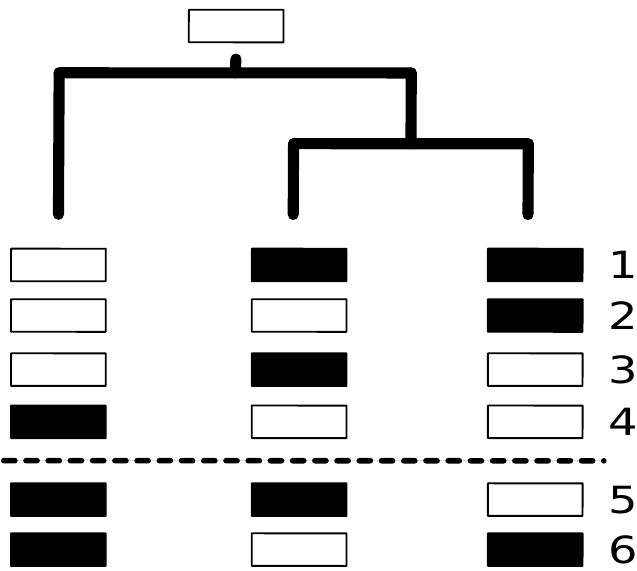}
	  \caption*{\small Partitions in 2 groups.}
	\end{figure}
    \end{minipage} \hfill
    \begin{minipage}{0.75\textwidth}
      \begin{itemize}
      \itemsep0em
	\item \NR{If we consider only unmarked partitions, then colorings 1, 2 and 3 induce the same partitions as, respectively, colorings 4, 5 and 6, and $N^{(\Tr)}_2=3$.}
	\item \NR{For marked partitions, fix the root state to an arbitrary color, for instance white, and consider the white group as marked. Then colorings 5 and 6 are not tree-compatible (they require two shifts). And although they induce the same partition, colorings 1 and 4 correspond to different marked partitions: each marks a different group of leaves. Therefore $M^{(\Tr)}_2=4$.}
    \end{itemize}
    \end{minipage}
    
  \end{example}
    
    \begin{proposition}[Computation of the Number of Equivalent Classes]\label{prop:number_eq_classes}
    Let $i$ be a node of tree $\Tr$, and $K\in \N^*$.\par
    If $i$ is a tip, then $N^{(\Tr_i)}_K = M^{(\Tr_i)}_K = \Ibb\{K = 1\}$.\par
    Else, if $i$ is a node with ${L_i}$ daughter nodes $(i_1, \cdots, i_{L_i})$, ${L_i}\geq 2$, then:
    \begin{equation}\label{eq:rec_num_eq_classes}
    \left\{
    \begin{aligned}
      N^{(\Tr_i)}_K & = \sum_{\substack{I\subset \intervalleentier{1}{L_i} \\ \card{I} \geq 2}}\quad \sum_{\substack{k_1+\dotsb + k_{L_i}=K+\card{I}-1 \\ k_1,\dotsc,k_{L_i} \geq 1}} \quad \prod_{l\in I} M^{(\Tr_{i_l})}_{k_l} \prod_{l\notin I} N^{(\Tr_{i_l})}_{k_l} + \sum_{\substack{k_1+\dotsb + k_{L_i}=K \\ k_1,\dotsc,k_{L_i} \geq 1}} \quad \prod_{l=1}^{L_i} N^{(\Tr_{i_l})}_{k_l} \\
      M^{(\Tr_i)}_K & = \sum_{\substack{I\subset \intervalleentier{1}{L_i} \\ \card{I} \geq 1}}\quad \sum_{\substack{k_1+\dotsb + k_{L_i}=K+\card{I}-1 \\ k_1,\dotsc,k_{L_i} \geq 1}} \quad \prod_{l\in I} M^{(\Tr_{i_l})}_{k_l} \prod_{l\notin I} N^{(\Tr_{i_l})}_{k_l}
    \end{aligned}
      \right.
    \end{equation}
    In the binary case, this relation becomes, if $i$ has two daughters $i_\ell$ and $i_r$:
    \begin{equation}\label{eq:rec_num_eq_classes_bin}
    \left\{
    \begin{aligned}
      N^{(\Tr_i)}_K & = \sum_{\substack{k_1+k_2 = K \\ k_1,k_2 \geq 1}} N^{(\Tr_{i_\ell})}_{k_1} N^{(\Tr_{i_r})}_{k_2} + \sum_{\substack{k_1+k_2 = K + 1\\ k_1,k_2 \geq 1}} M^{(\Tr_{i_\ell})}_{k_1} M^{(\Tr_{i_r})}_{k_2} \\
      M^{(\Tr_i)}_K & = \sum_{\substack{k_1+k_2 = K \\ k_1,k_2 \geq 1}} M^{(\Tr_{i_\ell})}_{k_1} N^{(\Tr_{i_r})}_{k_2} + N^{(\Tr_{i_\ell})}_{k_1} M^{(\Tr_{i_r})}_{k_2} + \sum_{\substack{k_1+k_2 = K + 1\\ k_1,k_2 \geq 1}} M^{(\Tr_{i_\ell})}_{k_1} M^{(\Tr_{i_r})}_{k_2}
    \end{aligned}
      \right.
    \end{equation}
    \end{proposition}
    
        \begin{proof}
        We will prove this proposition in the binary case, the general case being a natural extension of it.
    If $\Tr$ is a binary tree with $\Tr_{\ell}$ and $\Tr_{r}$ as left and right sub-trees, {one faces two situations when partitioning the tips in $K$ groups}:
    \begin{itemize}
    \item The left and right sub-trees have no group in common. Then, the number of groups in $\Tr$ is equal to the number of groups in its two sub-trees, and there are $\sum_{k_1+k_2=K} N^{(\Tr_{\ell})}_{k_1} N^{(\Tr_{r})}_{k_2}$ such partitions. This is the first term of the equation on $N^{(\Tr)}_K$ in~\ref{eq:rec_num_eq_classes_bin}.
    \item The left and right sub-trees have at least one group in common. Then, from the no homoplasy {assumption}, they have exactly one group in common: the ancestral state of the root. Suppose that this ancestral state is marked. Then it must be present in the two sub-trees, and there are  $\sum_{k_1+k_2=K+1} M^{(\Tr_{\ell})}_{k_1} M^{(\Tr_{r})}_{k_2}$ such partitions. This ends the proof of the formula on $N^{(\Tr)}_K$.
    \end{itemize}
    
    To get the formula on $M^{(\Tr)}_K$, we use the same kind of arguments. The second part of the formula is the same {as} the one for $N^{(\Tr)}_K$, and the first part corresponds to trees for which the marked partition is present in only one of the two sub-trees.
    \end{proof}
      The complexity of the algorithm described above is $O(2^L(K+L)^LLn)$.
   Note that $N^{(\Tr)}_K$ depends on the topology of the tree $\Tr$ in general. However, if the tree is binary, {a closed form solution of the recurrence relation~\ref{eq:rec_num_eq_classes}, which does not depend on the topology, exists}.
    
    \begin{corollary}[Closed Formula Binary Tree{s}]\label{expression_binary}
      For a rooted binary tree with $n$ tips, we have:
      \[
	N^{(\Tr)}_{K+1} = N^{(n)}_{K+1} = \card{\Sr_K^{PI}} = \binom{2n-2-K}{K}  \text{ and } M^{(\Tr)}_K = M^{(n)}_K = \binom{2n-K}{K-1}
      \]
    \end{corollary}
    
    The demonstration of this formula is not straightforward, and is based on a Vandermonde-like equality,  detailed in Appendix~\ref{supp:vandermonde}. The formula is then obtained using a strong induction on the number of tips of the tree.
    \PB{
    \begin{remark} \label{rk:sqrt_n} Note that, when $K$ is large compared to $\sqrt{n}$, the average number of configurations per equivalence class goes to infinity. This can be checked by comparing the total number of configurations $\binom{2n-1}{K-1}$ with the total number of classes $\binom{2n-K-1}{K-1}$. As a consequence, we only consider models for which $K < \sqrt{n}$ in the remainder.
    \end{remark}
    \begin{remark} This formula was already obtained in a different context in \citet{steel1992} (Proposition~1) and, with a slightly different formulation, in \citet[][Proposition~4.1.4]{semple2003}. In these works, the authors are interested in counting the \enquote{$r$-states convex characters on a binary tree}. Under the no-homoplasy assumption, this number can be shown to be equal to $\card{\Sr_{r-1}^{PI}}$.
    \end{remark}
    }


\subsection{Another Characterization of Parsimony}
The following proposition gives an alternative definition of parsimony under the no-homoplasy hypothesis using the linear formulation of the problem. It will be used for model selection in Section~\ref{sec:model_selection}. \NR{Its technical proof is postponed to Appendix~\ref{supp:proofs_identifiability}.}

    \begin{proposition}[Equivalence between parsimony and independence]\label{prop:parsimony_linear_independence}
     Let $\vect{m_Y}$ be a given mean vector, $\vect{m_Y}\in\R^{n}$, and $\vect{\Delta}$ a vector of shifts such that $\matr{T}\vect{\Delta} = \vect{m_Y}$, with $\matr{T}$ the tree matrix defined in Section~\ref{subsec:linear_model}. Under the no homoplasy {assumption}, the vector of shifts $\vect{\Delta}$ is parsimonious if and only if the corresponding column-vectors of the tree matrix $(T_i)_{i\in\Supp(\vect{\Delta})}$ are linearly independent.
      \end{proposition}
      

\section{Statistical Inference} \label{sec:inference}
  
  \subsection{Expectation Maximization}\label{subsec:EM}
  
    \paragraph{Principle} As shown in Section~\ref{subsec:latent_model}, both BM and OUsun models can be seen as incomplete data models. The Expectation Maximization \PB{algorithm} \citep[EM,][]{dempster1977} is a widely used algorithm for likelihood maximization of these kinds of models. It is based on the decomposition:
    \(
    \log p_{\vvect{\theta}}(\vect{Y}) = \Esp_{\vvect{\theta}}[\sachant{\log p_{\vvect{\theta}}(\vect{Z},\vect{Y})}{\vect{Y}}] - \Esp_{\vvect{\theta}}[\sachant{\log p_{\vvect{\theta}}(\sachant{\vect{Z}}{\vect{Y}})}{\vect{Y}}]
    \).
    Given an {estimate} $\vvect{\theta}^{(h)}$ of the parameters, we need to compute some moments of $p_{\vvect{\theta}^{(h)}}(\sachant{\vect{Z}}{\vect{Y}})$ (E step), and then find a new {estimate} $\vvect{\theta}^{(h+1)} = \argmax_{\vvect{\theta}}\Esp_{\vvect{\theta}^{(h)}}[\sachant{\log p_{\vvect{\theta}}(\vect{Z},\vect{Y})}{\vect{Y}}] $ (M step). The parameters are given for the BM and OUsun in subsection~\ref{subsec:probabilistic_model}. We assume here that the number of shifts $K$ is fixed.\par
    We only {provide} the main steps of the EM. {Additional} details can be found in Appendix~\ref{supp:EM}.
    
    \paragraph{E step} As $\vect{X}$ is Gaussian, the law of the hidden variables $\vect{Z}$ knowing the observed variables $\vect{Y}$ is entirely defined by its expectation and variance-covariance matrix, and can be computed using classical formulas for Gaussian conditioning. {The needed moments of $\sachant{\vect{Z}}{\vect{Y}}$ can also be computed using a procedure that is linear in the number of tips} (called \enquote{Upward-downward}) that takes advantage of the tree structure and bypasses inversion of the variance-covariance matrix \citep[see][for a similar algorithm]{lartillot2014}. 
    
%
    
    \paragraph{Complete Likelihood Computation} Using the model described in Section~\ref{subsec:latent_model}, we can use the following decomposition of the complete likelihood:
    \[
    p_{\vvect{\theta}}(\vect{X}) = p_{\vvect{\theta}}(X_1)\prod_{j=2}^{m+n}p_{\vvect{\theta}}\left(\sachant{X_j}{X_{\pa(j)}}\right)
    \]
    Each term of this product is then known, and we easily get $\Esp_{\vvect{\theta}^{(h)}}[\sachant{\log p_{\vvect{\theta}}(\vect{Z},\vect{Y})}{\vect{Y}}]$.

    \paragraph{M step} The difficulty comes here from the discrete variables (location of shifts on the branches). The maximization is exact for the BM but we only raise the objective function for the OUsun, hence computing a Generalized EM \citep[GEM, see][]{dempster1977}. This stems from the independent increment nature of the BM: shifts only affect $p_{\vvect{\theta}}\left(\sachant{X_j}{X_{\pa(j)}}\right)$ on the branches where they occur and the maximization reduces to finding the $K$ highest components of a vector, which has complexity $O(n + K\log(n))$. By contrast, OUsun has autocorrelated increments: shifts affect $p_{\vvect{\theta}}\left(\sachant{X_j}{X_{\pa(j)}}\right)$ on the branches where they occur and on all subsequent branches. Maximization is therefore akin to segmentation on a tree, which has complexity $O(n^K)$.

    \paragraph{Initialization} Initialization is always a crucial step when using an EM algorithm. Here, we use the linear formulation~\ref{eq:BM_lin_def_X} or~\ref{eq:OU_lin_def}, and initialize the vector of shifts \PB{using} a Lasso regression. The selection strength $\alpha$ is initialized using {pairs} of tips {likely to be in the same group.} 
  
  \subsection{Model Selection}\label{sec:model_selection}
  
  \paragraph{Model Selection in the iid Case with Unknown Variance} Model selection in a linear regression setting has received a lot of attention over the last few years. In \citet{baraud2009}, the authors developed a non-asymptotic method for model selection in the case where the errors are independent and identically distributed (iid), with an unknown variance. In the following, we first recall their main results, and then {adapt it} to our setting of non-independent errors.\par
   
    We assume that we have the following model of \emph{independent} observations:
    \[
     \vect{Y}' = \vect{s}' + \gamma \vect{E}' \quad \text{ with } \quad  \vect{E}'\sim\Normal(0, \matr{I_n})
    \]
    and we define a collection $\Sr' = \{S'_{\eta}, \eta \in \Mr\}$ of linear subspaces of $\R^n$ that we call \emph{models}, and that are indexed by a finite or countable set $\Mr$. For each $\eta\in\Mr$, we denote by $\vect{\hat{s}}'_\eta = \Proj_{S'_\eta}\vect{Y}'$ the orthogonal projection of $\vect{Y}'$ on $S'_{\eta}$, that is a least-square estimator of $\vect{s}'$, and $\vect{s}'_\eta = \Proj_{S'_\eta}\vect{s}'$ the projection of $\vect{s}'$.
    
        We extract from \citet{baraud2009} the following theorem, that bounds the risk of the selected estimator, and provides us with a non-asymptotic guarantee. It relies on a penalty depending on the $\EDkhi$ function, as defined below:
        
    \begin{definition}[\citet{baraud2009}, Section 4, definitions 2 and 3]
    Let $D$, $N$ be two positive integers, and $X_D$, $X_N$ be two independent $\chi^2$ random variables with degrees of freedom $D$ and $N$ respectively. For $x\leq 0$, define
    \[
      \Dkhi[D, N, x] = \frac{1}{\Espe{X_D}}\Espe{\left(X_D - x\frac{X_N}{N}\right)_+}
    \]
    And define $\EDkhi[D,N,q]$ as the unique solution of the equation \(\Dkhi[D, N, \EDkhi[D, N, q]] = q\) (for $0 < q \leq 1$).
    \end{definition}
    
    \begin{theorem}[\citet{baraud2009}, Section 4, theorem 2 and corollary 1]\label{th:model_selection_baraud2009}
    In the setting defined above, let $D_{\eta}$ be the dimension of $S'_\eta$, and assume that $N_\eta = n - D_\eta \geq 2$ for all $\eta \in \Mr$.
    Let $\Lr = \{L_\eta\}_{\eta\in\Mr}$ be some family of positive numbers such that
      \(
  \Omega' = \sum_{\eta \in \Mr} (D_\eta + 1) e^{-L_\eta} < +\infty
      \), 
      and assume that, for $A>1$, 
      \[
      \pen(\eta) = \pen_{A, \Lr}(\eta) = A \frac{N_\eta}{N_\eta - 1} \EDkhi[D_\eta+1, N_\eta-1, e^{-L_\eta}]
      \]
      Take $\hat{\eta}$ as the minimizer of the criterion:
    \(
    \hat{\eta} = \argmin_{\eta\in\Mr} \norm{\vect{Y}' - \vect{\hat{s}}'_\eta}^2 \left( 1 + \frac{\pen(\eta)}{N_\eta}\right)
     \label{eq:crit_LQ}
    \).\\
Then, assuming that $N_\eta\geq7$ and $\max(L_\eta, D_\eta) \leq \kappa n$ for any $\eta\in\Mr$, with $\kappa < 1$,
     the following non-asymptotic bound holds:
	\[
	\Espe{\frac{\norm{\vect{s}' - \vect{\hat{s}}'_{\hat{\eta}}}^2}{\gamma^2}} \leq C(A, \kappa)\left[ \inf_{\eta \in \Mr} \left\{ \frac{\norm{\vect{s}' - \vect{s}'_\eta}^2}{\gamma^2} + \max(L_\eta, D_\eta) \right\} + \Omega'\right]
      \]
    \end{theorem}
    
    The penalty used here ensures an oracle inequality: in expectation, the risk of the selected estimator is bounded by the risk of the best possible estimator of the collection of models, up to a multiplicative constant, and a residual term that depends on the dimension of the oracle model. Note that if the collection of models is poor, such an inequality has low value. We refer to \citet{baraud2009} for a more detailed discussion of this result.
    
    
    \paragraph{Adaptation to the Tree-Structured Framework} We use the linear formulation described in~\ref{subsec:linear_model}, and assume that we are in the OUsun model (this procedure would also work for a BM with a deterministic root). Then, if $\matr{V}$ is a matrix of size $n$, with $V_{ij} = e^{-\alpha d_{ij}}, \forall (i,j)\in\intervalleentier1n^2$, we have: 
    \[
    \vect{Y} = \matr{T}\matr{W}(\alpha)\vect{\Delta} + \gamma \vect{E} = \vect{s} + \gamma \vect{E} \quad \vect{E} \sim \Normal(0, \matr{V})
    \]
    We assume that $\alpha$ is fixed, so that the design matrix $\matr{T}\matr{W}(\alpha)$ and the structure matrix $\matr{V}$ are known and fixed.
A \emph{model} is defined here by the position of the shifts on the branches of the tree, i.e.\ by the non-zero components of $\vect{\Delta}$ (with the constraint that the first component, the root, is always included in the model). We denote by $\Mr = \bigcup_{K = 0}^{p-1} \Sr_K^{PI}$ the set of allowed (parsimonious) allocations of shifts on branches (see Section~\ref{subsec:complexity_model}), $p$ being the maximum allowed dimension of a model. From proposition~\ref{prop:parsimony_linear_independence}, for $\eta \in \Mr$, the columns vectors $\matr{T}_\eta$ are linearly independent, and the model $S_\eta = \Span(\vect{T_i}, i\in\eta)$ is a linear sub-space of $\R^{n}$ of dimension $D_\eta = \card{\eta} = K_{\eta} + 1$, $K_\eta$ being the number of shifts in model $\eta$. Note that as $\matr{W}(\alpha)$ is diagonal invertible, it does not affect the definition of the linear subspaces. The set of models is then $\Sr = \{S_\eta, ~ \eta\in\Mr\}$.\par
We define the Mahalanobis norm associated to $\matr{V}^{-1}$ by:
\(
 \mahanorm{\vect{R}}{V} = \vect{R}^T \matr{V}^{-1} \vect{R},~ \forall \vect{R} \in \R^{n}
\).

The projection on $S_\eta$ according to the metric defined by $\matr{V}^{-1}$ is then:
    \[
     \vect{\hat{s}}_\eta = \Proj^{\matr{V}^{-1}}_{S_\eta}(\vect{Y}) = \argmin_{\vect{a} \in S_\eta} \mahanorm{\vect{Y} - \vect{a}}{V}^2 \quad \text{ and }\quad \vect{s}_\eta = \Proj^{\matr{V}^{-1}}_{S_\eta}(\vect{s})
    \]
    For a given number of shifts $K$, we define the best model with $K$ shifts as the one maximizing the likelihood, or, equivalently, minimizing the least-square criterion for models with $K$ shifts:
    \[
    \vect{\hat{s}}_K = \argmin_{\eta \in \Sr, \card{\eta} = K+1} \mahanorm{\vect{Y} - \vect{\hat{s}}_\eta}{V}^2
    \]
    The idea is then to slice the collection of models by the number of shifts $K$ they employ. Thanks to the EM algorithm above, we are able to select the best model in such a set. The problem is then to select a reasonable number of shifts. To compensate the {increase in} the likelihood due to over-fitting, using the model selection procedure described above, we select $K$ using {the following} penalized criterion{:}
    \begin{equation}\label{eq:crit_MC}
    \Crit_{LS}(K) = \mahanorm{\vect{Y} - \vect{\hat{s}}_K}{V}^2 \left( 1 + \frac{\pen(K)}{n-K-1}\right)
    \end{equation}
    As noted in \citet{baraud2009}, the previous criterion can {equivalently} be re-written in term of likelihood, as:
    \begin{equation}\label{eq:crit_LL}
    \Crit_{LL}(K) = \frac{n}{2} \log\left(\frac{\mahanorm{\vect{Y} - \vect{\hat{s}}_K}{V}^2}{n}\right) + \frac12 \pen'(K)
    \end{equation}
    with     \( \pen'(K) = n\log\left(1 + \frac{\pen(K)}{n-K-1}\right) \). As we use maximum-likelihood estimators, we chose this formulation for the implementation.
The following proposition then holds:

        \begin{proposition}[Form of the Penalty and guaranties ($\alpha$ known)]\label{prop:model_selection}
    Let $\Lr = \{L_K\}_{K\in\intervalleentier0{p-1}}$, with $p \leq \min\left(\frac{\kappa n}{2 + \log(2) + \log(n)}, n-7\right)$, the maximum dimension of a model, with $\kappa < 1$, and:
    \begin{equation}\label{eq:L_K}
    L_K = \log\card{\Sr_{K}^{PI}} + 2\log(K+2), \forall K\in\intervalleentier0{p-1}
    \end{equation}
    Let $A>1$ and assume that 
    \[
      \pen_{A, \Lr}(K) = A \frac{n - K - 1}{n - K - 2} \EDkhi[K, n - K - 2, e^{-L_{K}}]
    \]
    Suppose that $\hat{K}$ is a minimizer of~\ref{eq:crit_MC} or~\ref{eq:crit_LL} with this penalty. Then:
     \[
      \Espe{\frac{\mahanorm{\vect{s} - \vect{\hat{s}}_{\hat{K}}}{V}^2}{\gamma^2}} \leq C(A, \kappa)\inf_{\eta \in \Mr} \left\{ \frac{\mahanorm{\vect{s} - \vect{s}_\eta}{V}^2}{\gamma^2} + (K_\eta + 2)\left(3 +\log(n)\right) \right\}
     \]
     with $C(A, \kappa)$ a constant depending on $A$ and $\kappa$ only.
    \end{proposition}
  
  The proof of this proposition can be found in Appendix~\ref{supp:proof_model_selection}. It relies on theorem~\ref{th:model_selection_baraud2009}, adapting it to our tree-structured observations.
  
 \begin{remark}
  With this oracle inequality, we can see that we are missing the oracle by a $\log(n)$ term. This term is known to be unavoidable, see \citet{baraud2009} for further explanations.
 \end{remark}
    
    \begin{remark}
Note that the chosen penalty {may} depend on the topology of the tree through the term $\card{\Sr_{K}^{PI}}$ (see Section~\ref{subsec:complexity_model}).
    \end{remark}
    
    \begin{remark}
    The penalty involves a constant $A>1$, that need{s} to be chosen by the user. Following \citet{baraud2009} \PB{who tested a series of values}, we fixed this constant to $A=1.1${.}
    \end{remark}

\section{Simulations Studies} \label{sec:simulations}

  \subsection{Simulations Scheme}
%
    We tested our algorithm on data simulated according to an OUsun, with varying parameters. The simulation scheme is inspired \PB{by the work of} \citet{uyeda2014}.
    We first generated three distinct trees with, respectively, $64, 128$ and $256$ tips, using a pure birth process with {birth rate} $\lambda = 0.1$. 
The tree heights were scaled to one, and their topology and branch lengths were fixed for the rest of the simulations. 
    We then used a star-like simulation study scheme, fixing a base scenario, and exploring the space of parameters one direction at {the} time. The base scenario was taken to be relatively \enquote{easy}, with $\beta_1 = 0$ (this parameter was fixed for the rest of the simulations), $\alpha_b = 3$ (i.e $t_{1/2,b} = 23\%$), $\gamma^2_b = 0.5$ and $K_b = 5$.
    The parameters then varied in the following ranges: the phylogenetic half life $t_{1/2} = \ln(2)/\alpha$ took $11$ values in $\intervalleff{0.01}{10}$;
     the root variance $\gamma^2 = \frac{\sigma^2}{2\alpha}$ took $9$ values in $\intervalleff{0.05}{25}$;
     the number of shifts $K$ took $9$ values in $\intervalleff{0}{16}$ (see Figures~\ref{fig:ll_t_gamma}-\ref{fig:beta_K_ARI} for the exact values taken).
     The problem was all the more \emph{difficult} that $\gamma^2$, $t_{1/2}$ or $K$ were large.\par
    For each simulation, the $K$ shifts were generated in the following way. First, their values were drawn according to a mixture of two Gaussian distributions, $\Normal(4,1)$ and $\Normal(\ensuremath{-4},1)$, in equal proportions. The mixture was chosen to avoid too many shifts of small amplitude. Then, their positions were chosen to be balanced: we first divided the tree in $K$ segments of equal heights, and then randomly drew in each segment an edge where to place a shift. {We only kept} parsimonious allocations{.}\par

    Each of these configurations was repeated $200$ times, leading to $16200$ simulated data sets. An instance of a tree with the generated data is plotted in Figure~\ref{simus:eq_shifts}.

        \begin{figure}[!ht]

\includegraphics[width=\linewidth]{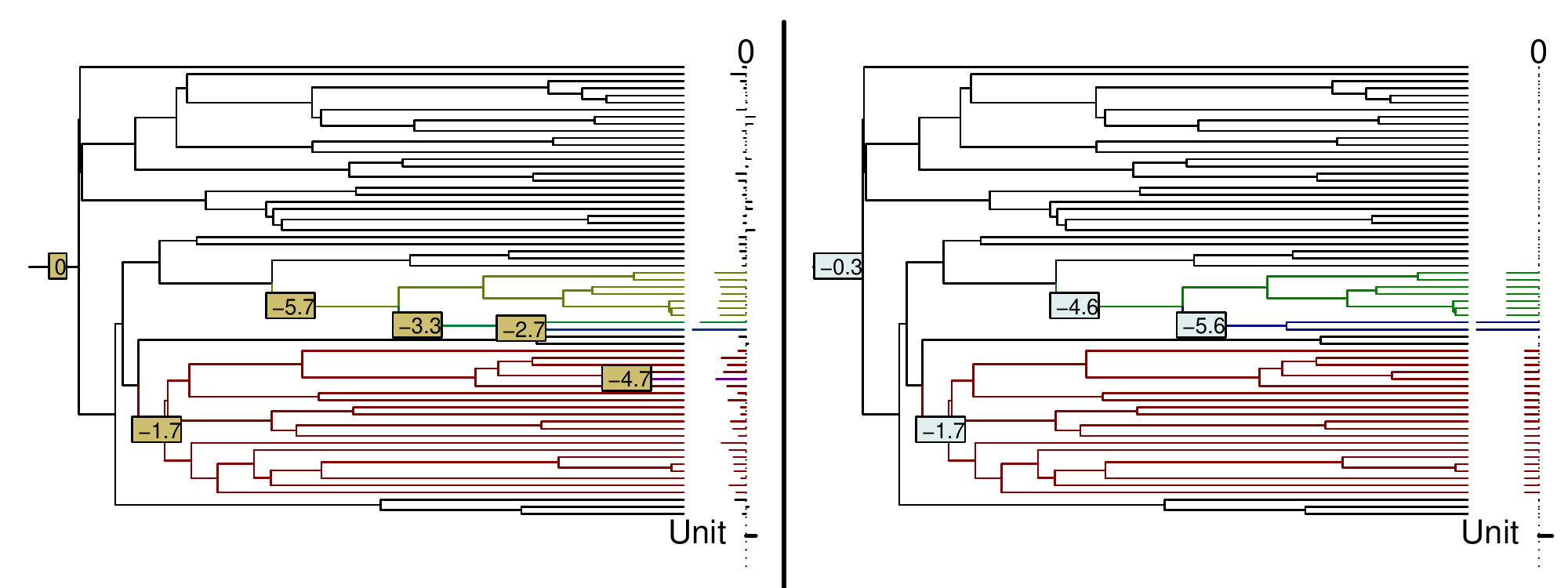} 

\caption{\emph{Left}: Simulated configuration (with $t_{1/2} = 0.75$, $\gamma^2 = 0.5$ and $K = 5$). The shifts positions and values are marked on the tree. The value of the character generated (positive or negative) is represented on the right. The colors \PB{of the branches} correspond to the true regimes, black being the ancestral state. 
\emph{Right}: One of the three equivalent allocations of shifts for the model inferred from the data, with corresponding vector of mean tip values. Shifts \PB{not recovered are located on pendant edges, and} have low influence on the data. The two other equivalent allocations can be easily deducted from this one.}\label{simus:eq_shifts}
\end{figure}
    
    \subsection{Inference Procedures}
    For each generated dataset, we ran our EM procedure with fixed values of $K \in \intervalleentier{0}{\floor{\sqrt{n}}}$, $n$ being the number of tips of the tree. Remark that for $n = 64$, $\floor{\sqrt{n}} = 8$, and we have no hope of detecting true values of $K$ above $8$ \PB{(see remark~\ref{rk:sqrt_n} for an explanation of the bound in $\sqrt{n}$)}. The number of shifts $K_s$ was chosen thanks to our penalized criterion, and we kept inferences corresponding to both $K_s$ and the true number $K_t$.\par
    We ran two sets of estimations for $\alpha$ either known or estimated.
    The computations took respectively $66$ and $570$ (cumulated) days of CPU time. This amounts to a mean computational time of around $6$
    minutes ($367$ seconds) for one estimation when $\alpha$ is fixed, and $52$
    minutes ($3137$ seconds) when $\alpha$ is estimated, \PB{with large differences between easy and difficult scenarios.}
 


    \subsection{Scores Used} The convergence of the EM algorithm was assessed through the comparison of the likelihood of the true and estimated parameters, and the comparison of mean number of EM steps needed when $\alpha$ is fixed or estimated. {The quality of the estimates of} $\beta_1$, $t_{1/2}$ and $\gamma^2$ {was assessed using the coefficient of variation}.
The model selection procedure was evaluated by comparing the true number of shifts with the estimated one, {which} should be lower. We do not expect to find the exact number as some shifts, {which} are too small or too close of the tips, cannot be detected.
To evaluate the quality of the clustering of the tips, the only quantity we can observe, we used the Adjusted Rand Index \citep[ARI,][]{hubert1985} between the true clustering of the tips, and the one induced by the estimated shifts. The ARI \PB{is proportional to the number of concordant pairs in two clusterings and} {has maximum value of $1$ \PB{(for identical clusterings)} and expected value of $0$ {(for random clusterings)}. Note that this score {is conservative as} shifts of small intensity, {which} are \PB{left aside} by our model selection procedure, produce \enquote{artificial} groups that cannot be reconstructed. 

    
    \subsection{Results} The selection strength is notoriously difficult to estimate, with large ranges of values giving similar behaviors \citep[see][]{thomas2014}. We hence first analyse the impact of estimating $\alpha$ on our estimations, showing that the main behavior of the algorithm stays the same. Then, we study the shifts reconstruction procedure.

    \begin{figure}[h!]
\begin{center}

\includegraphics[width=\linewidth]{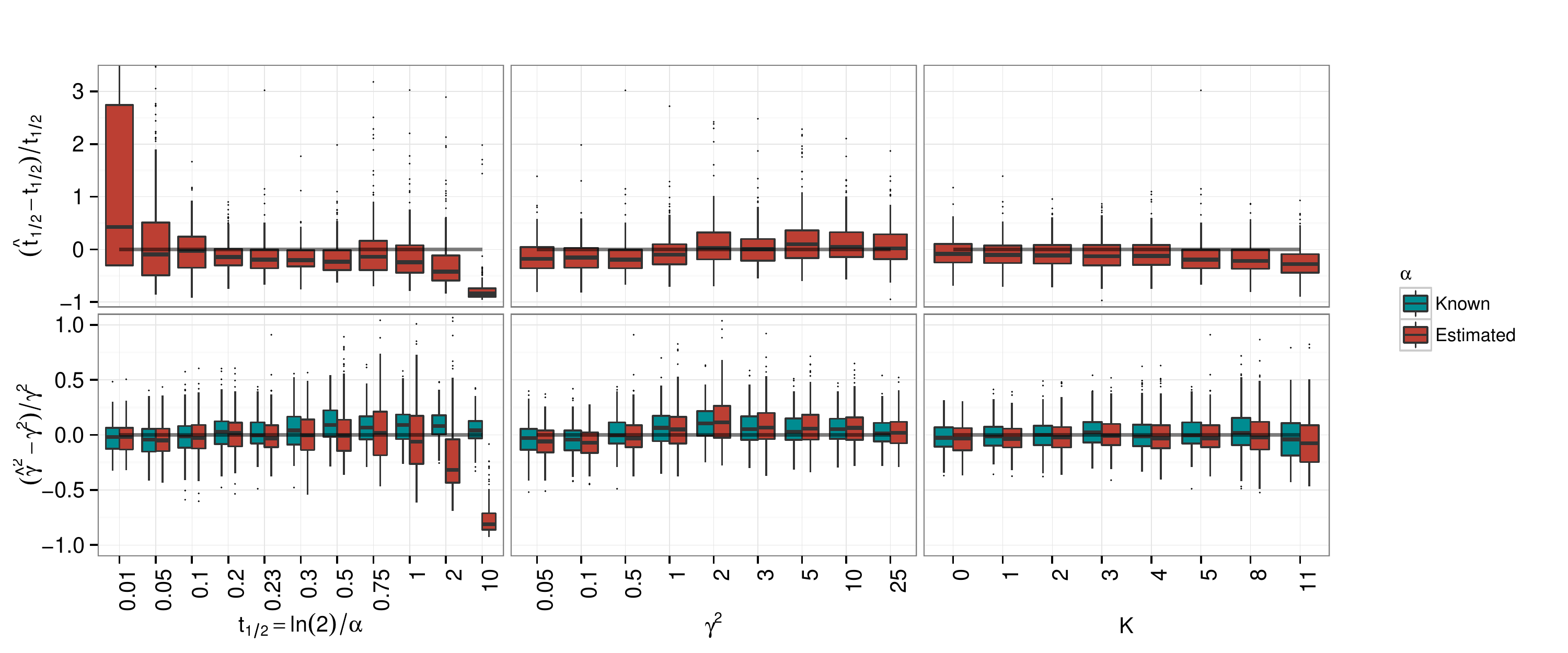} 

\caption{Box plots over the $200$ repetitions of each set of parameters, for the
phylogenetic half-life (\PB{top}) and root variance (bottom) with $K$ estimated, and $\alpha$ fixed \PB{ to its true value (blue}) or estimated (\PB{red}), on a tree with $128$ taxa.
\PB{For better legibility, the $y$-axis of these two rows were re-scaled, omitting some outliers (respectively, for $t_{1/2}$ and $\gamma^2$, $0.82\%$ and $0.46\%$ of points are omitted). The whisker of the first box for $t_{1/2}$ goes up to $7.5$.}}\label{fig:ll_t_gamma}
\end{center}

\begin{center}

\includegraphics[width=\linewidth]{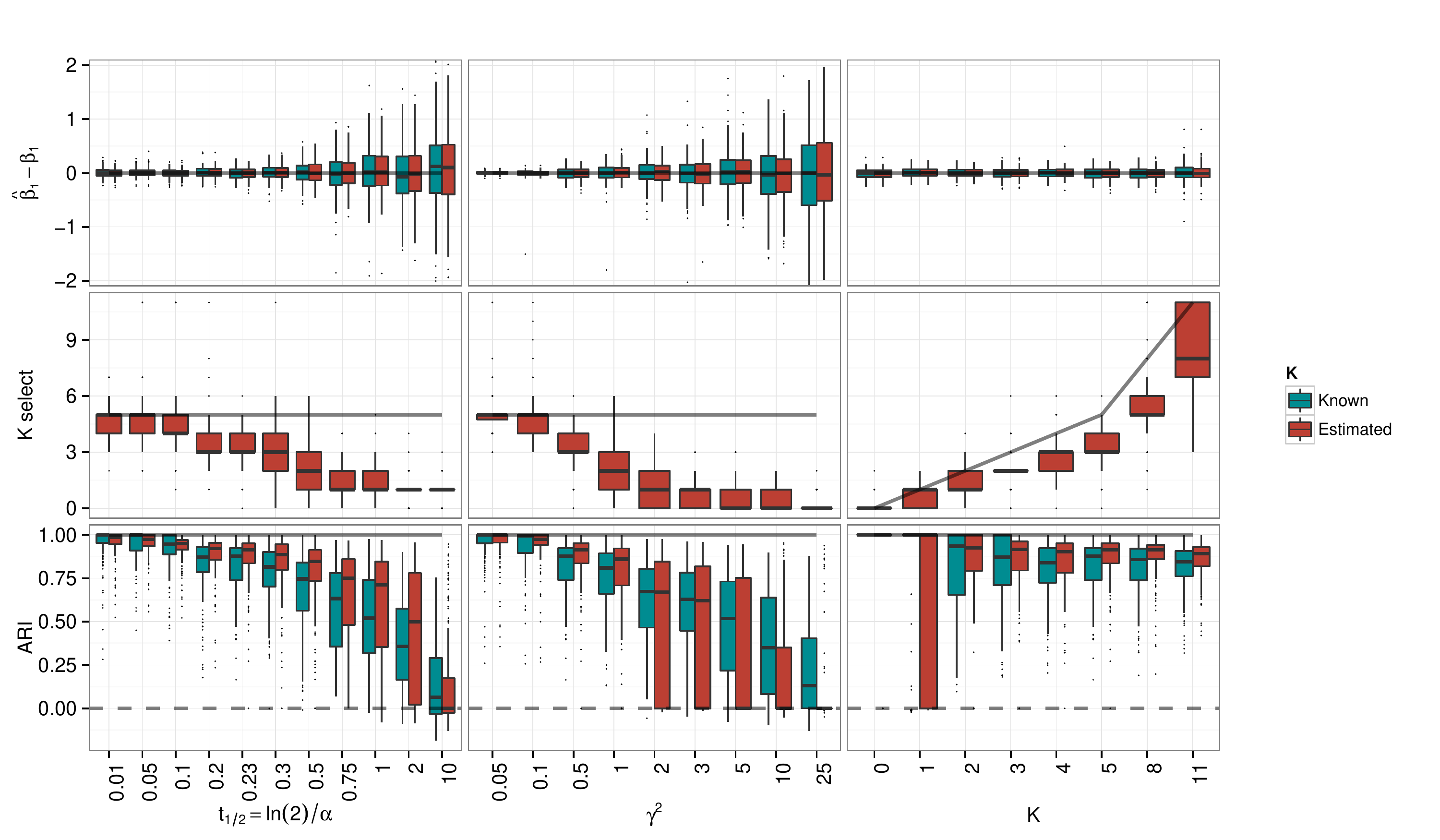} 

\caption{Same for $\beta_1$ (top), the number of shifts (middle) and ARI (bottom), with $\alpha$ estimated, and $K$ fixed \PB{to its true value (blue}) or estimated (\PB{red}). 
\PB{As previously, the $y$-axis of the first row ($\beta_1$) was re-scaled, omitting some outliers ($1.38\%$ of points are omitted).}}
\label{fig:beta_K_ARI}
\end{center}

\end{figure}

\paragraph{Convergence and Likelihood}

{For $\alpha$} known, all estimations converged in less than $49$ iterations, with a median number of $13$ iterations. {For $\alpha$} estimated, the number of iterations increased greatly, with a median of $69$, and a fraction of estimations (around $3.3$\%) that reached the maximum allowed number (fixed at $1000$ iterations) without converging. Unsurprisingly, the more difficult the problem, the more iterations were needed.
    The log-likelihoods of the estimated parameters \PB{are close} to the true ones, even when $\alpha$ is estimated \PB{(see supplementary Figure \ref{fig:ll_t_gamma_rev} in Appendix~\ref{supp:SimulationAnalysis}, first row)}. 
    
    \paragraph{Estimation of continuous parameters}
    Figure~\ref{fig:ll_t_gamma} (first \PB{row}) shows that we tend to slightly over-estimate $\alpha$ in general. The estimation is particularly bad for large values of $\alpha$ (with a high variance on the result, see first box of the \PB{row}), and low values of $\alpha$. In this regime, the model is \enquote{over-confident}, as it finds a higher selection {strength} than the {real} one {and therefore a} smaller variance (second \PB{row} of Figure~\ref{fig:ll_t_gamma}).
    For smaller and bigger trees, the estimators behave {in} the same way, but with degraded or improved values, as expected. We also note that taking the true number of shifts instead of the estimated one slightly degrades our estimation of these parameters (see supplementary Figure~\ref{fig:ll_t_gamma_rev} in Appendix~\ref{supp:SimulationAnalysis}).
    The estimation of $\beta_1$ is not affected by the knowledge of $\alpha$ or $K$ (see Figure~\ref{fig:beta_K_ARI}, first \PB{row}), and only has an increased variance for more difficult configurations.
    In the remainder, we only show results obtained for estimated $\alpha$ as estimating $\alpha$ does not impact ARI, $\hat{K}$ and $\hat{\beta}_0$ (see supplementary Figure~\ref{fig:beta_K_ARI_rev} in Appendix~\ref{supp:SimulationAnalysis}).

    \paragraph{Estimation of the number of shifts} The way shifts were drawn ensures that they are not too small in average, and that they are located all along the tree. Still, some shifts have a very small influence on the data, and are hence hard to detect (see Figure~\ref{simus:eq_shifts}). The selection model procedure almost always under-estimates the number of shifts, except in very favorable cases (Figure~\ref{fig:beta_K_ARI}, second \PB{row}). This behavior is nonetheless expected, as allowing more shifts does not guarantee that the right shifts will be found (see supplementary Figures~\ref{simus:fpr} and~\ref{simus:sensitivity} in Appendix~\ref{supp:SimulationAnalysis}). 

    \paragraph{Clustering of the tips} {T}he ARI tends to be degraded for small values of $\alpha$ or high variance, but remains positive {(Figure~\ref{fig:beta_K_ARI}, third \PB{row})}. When only one shift occurs, the ARI is very unstable, but for any other value of $K$, it stays quite high. Finally,  knowing the {number of shifts} does not improve the ARI. 
    
    
    \paragraph{Equivalent Solutions}

When $\alpha$ {and $K$ are both} estimated, only $5.8\%$ of the configurations have $2$ or more equivalent solutions. One inferred {configuration} with three equivalent solutions is presented Figure~\ref{simus:eq_shifts}. 

    \paragraph{Comparison with \printR{bayou}} As mentioned above, our simulation scheme, although not completely equivalent to the scheme used in \citet{uyeda2014}, is very similar, so that we can compare our results with theirs. The main differences lies in the facts that we took a grid on $\gamma^2 = \sigma^2/2\alpha$ instead of $\sigma^2$, and that we took shifts with higher intensities, making the detection of shifts easier. 
    We can see that we get the same qualitative behaviors for our estimators, with the selection strength $\alpha$ over or under estimated, respectively, in small or large values regions. The main difference lies in the estimation of the number of shifts. Maybe because of the priors they used ($K \sim \text{Conditional Poisson} ( \lambda = 9, K_{max} = n/2)$), they tend to estimate similar numbers of shifts (centered on $9$) for any set of parameters. In particular, while our method seems to be quite good at detecting situations where there are no shifts at all, theirs seems unable to catch these kind of configurations, despite the fact that their shifts have low intensity, leading to a possible over-fitting of the data.
    
    \PB{Overall, the behavior of the algorithm is quite satisfying. Our model selection procedure avoids over-fitting, while recovering the correct clustering structure of the data. It furthermore allows for a reasonable estimation of the continuous parameters, except for $\alpha$ which is notoriously difficult to estimate.}

\section{Case Study: Chelonian Carapace Length Evolution} \label{sec:chelonia}


\begin{figure}[!ht]
  \begin{center}

\includegraphics[width=0.8\linewidth]{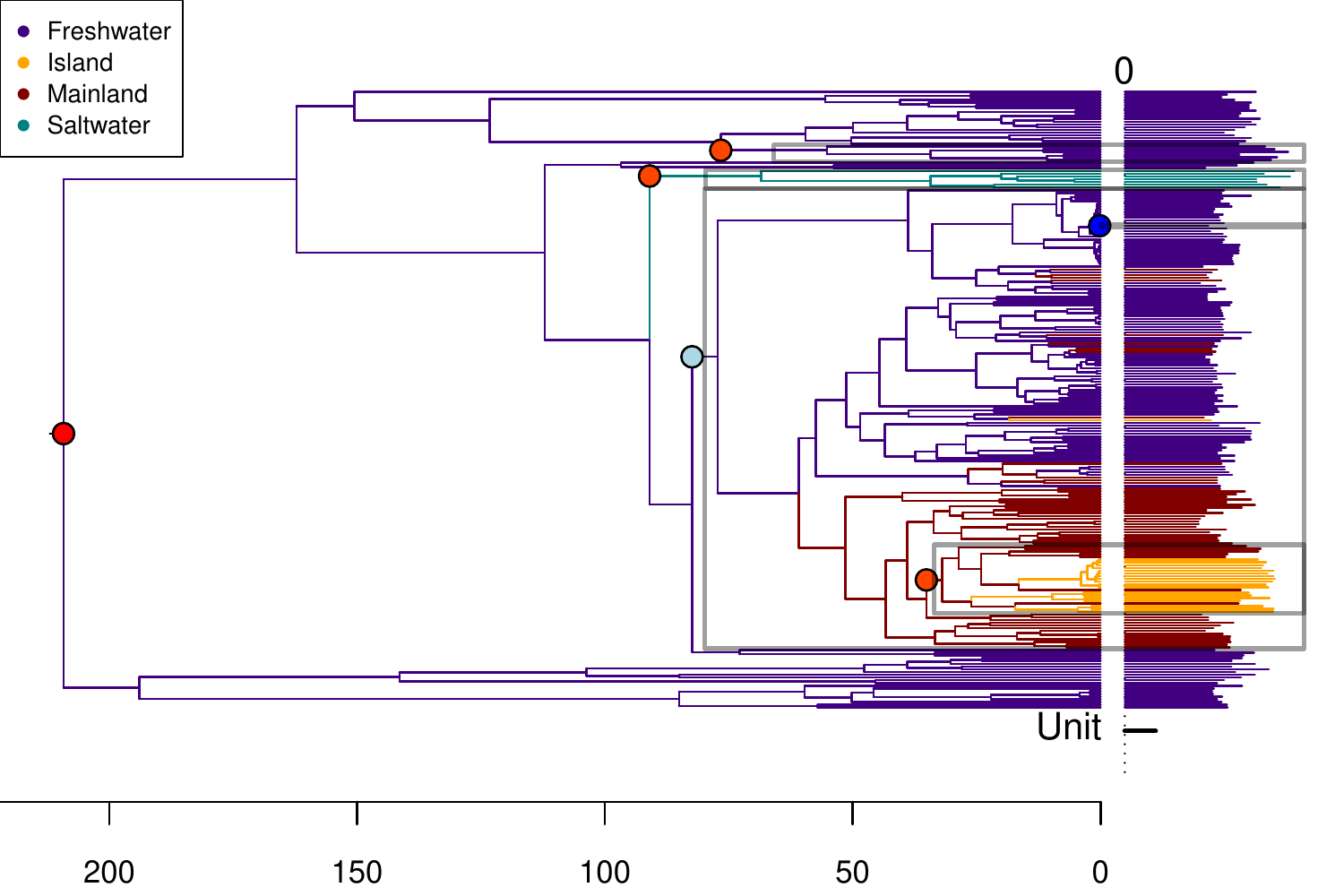} 

\caption{Phylogenetic tree of the Chelonians. {L}og-transformed traits values are represented on the right. {Branch colors} represent the habitats. The shifts found by our EM algorithm are shown {as circles}, with a color indicating the value of the shift, from blue (negative) to red (positive). Boxes highlight the groups {induced} by the shifts. The $x$-scale is in million years.}\label{fig:chelonia_sol}
\end{center}
\end{figure}

\subsection{Description of the Dataset}
Extant species of the order Testudines, or Chelonii, are turtles and tortoises, living all across the globe, and exhibiting a wide variation in body size, from the small desert speckled tortoise (\emph{Homopus signatus}, $10$ cm), to the large marine leatherback sea turtle (\emph{Dermochelys coriacea}, $244$ cm). In order to test the hypothesis of island and marine gigantism, that could explain the extreme variations observed, \citet{jaffe2011} compiled a dataset containing a measure of the carapace length for $226$ species, along with a phylogenetic tree of theses species, spanning  $210$ million years (my) (see Figure~\ref{fig:chelonia_sol}). They assigned each species to one of four habitats: mainland-terrestrial, freshwater, marine and island-terrestrial. Then, testing several fixed regimes allocations on the branches of the tree using the method described in \citet{butler2004}, they found the {best} support in favor of {a} \enquote{OU2} model that assigned one regime to each habitat. Following \citet{uyeda2014}, we will refer to this model as the \enquote{OU$_{habitat}$} model. Note that this model is ambiguously defined, as it requires to assign a habitat to each ancestral species. Using proposition~\ref{prop:size_eq_class}, we found that there were $48$ equivalent parsimonious ways of doing so that respect the habitats observed at the tips of the tree. One of these {habitat reconstruction} is presented Figure~\ref{fig:chelonia_sol}.
  
\subsection{Method} We used the version of the dataset embedded in the package \citeR{geiger}, that contains a phylogenetic tree and a vector of log-carapaces lengths. \PB{The corresponding habitats are reported in} the Appendix of \citet{jaffe2011}.\par
  We ran our algorithm with a number of shifts going from $0$ to $20$. \NR{Rather than estimating $\alpha$ directly within the EM as we did for the simulations, we took $\alpha$ varying on a grid, taking $6$ values regularly spaced between $0.01$ and $0.1$, 
  but fixed for each estimation. We found that this approach, although computationally more intensive, gave better results.} These $6\times 20 = 120$ estimations took around $2$
  hours of CPU time. For each number of shifts, going again from $0$ to $20$, we kept the solution with the maximal likelihood, and we applied the model selection criterion to them. This method gave a solution with $5$ shifts, and a selection strength of $0.06$ (i.e.\ $5.5$ \% of the total height of the tree). Using a finer grid for $\alpha$ gives highly similar {results}, allocating shifts {to} the same edges. These last estimations are given below.\par

  \subsection{Results}
  Our method selected a solution with $5$ shifts, a rather strong selection strength ($t_{1/2} = 5.4\%$ of the tree height), and a rather low root variance ($\gamma^2 = 0.22$, see table~\ref{table:chelonia}, first column).
  Two of those shifts are closely related to the habitats defined in \citet{jaffe2011} (see Figure~\ref{fig:chelonia_sol}). The ancestral optimal value, that applies here to two clades of freshwater turtles, is estimated to be around $38$ cm. A small decrease in size for a large number of mainland and freshwater turtles is found (optimal value $24$ cm). Marine {turtles} (super-family Chelonioidea) are found to have an increased carapace length (with an optimal value of $130$ cm), as well as a clade containing soft-shell tortes (family Trionychidae, optimal size $110$ cm), and a clade containing almost all island tortoises, including several sub-species of Gal\'apagos tortoises (\emph{Geochelone nigra}). Only the Ryukyu black-breasted leaf turtle (\emph{Geoemyda japonica}), endemic to the Ryukyu Islands in Japan, and distant on the phylogenetic tree, is not included in this group. Note that the group {also} contains some mainland tortoises of the genus Geochelone, that are closely related to Gal\'apagos tortoises. This is typical of our method: it constructs groups that are both phenotypically and phylogenetically coherent. Finally, one species is found to have its own group, the black-knobbed map turtle (\emph{Graptemys nigrinoda}), with a very low optimal value of $\ensuremath{1.4\times 10^{-20}}$ cm, for a measured trait of $15$ cm. The fact that the shift has a very high negative value ($\ensuremath{-49}$ in log scale) is just an artifact due to the actualization factor on a very small branch ($0.18$ my, for an inferred phylogenetic half-life of $11$ my). This is a rather unexpected choice of shift location. When considering the linear model as transformed by the cholesky matrix of the variance to get independent errors (as in the proof of proposition~\ref{prop:model_selection}), we find a leverage of $0.94$, indicating that this species trait behaves in the transformed space as an outsider.\par
  
  \subsection{Comparison with other methods}\label{sec:chelonia_comparison}

In order to compare our results to previously published ones, we reproduced some of the analysis already conducted on this dataset. \NR{We hence ran the methods described in \citet{jaffe2011} (using the \printR{R} package \printR{OUwie}, with fixed positions for the shifts), \citet{uyeda2014} (implemented in package \printR{bayou}), \citet{SURFACE} (package \printR{SURFACE}), and \citet{hoane2014} (function \printR{OUshifts} in package \printR{phylolm}). See Section~\ref{supp:chelonia} in Appendix~\ref{supp:SimulationAnalysis} for more details on these methods and the parameters we used.}

  The shifts allocated on the tree by methods \printR{bayou}, \printR{SURFACE} and \printR{OUshifts} are presented on Supplementary Figure~\ref{fig:chelonia_sol_bayou} (Appendix~\ref{supp:SimulationAnalysis}). We can see that $3$ among the most strongly supported shifts in the posterior distribution given by \printR{bayou}, as well as \PB{some among the oldest shifts found by \printR{SURFACE} and \printR{OUshifts}}
  are similar to the ones found by our method. The \printR{bayou} method finds equal support for many shifts, all over the tree, and the median of the posterior distribution is $17$ shifts, which is pretty close to the mode of the prior put on the number of shifts ($15$). The \printR{SURFACE} \PB{and \printR{OUshifts} methods select respectively $33$ and $8$ shifts}, including many on pendant edges, that are not easily interpretable. The backward step of \printR{SURFACE} allowed to merge the regimes found for marine turtles and soft-shell tortoises that our method found to have very similar optimal values.
  The results of the five methods are summarized Table~\ref{table:chelonia}. Note that these models are not nested, due to the status assigned to the root, and to the possible convergences. 
  
  Compared to step-wise heuristics, our integrated maximum likelihood based approach allows us to have a more \enquote{global} view of the tree, and hence to select a solution that accounts better for the global structure of the trait distribution. Thanks to its rigorous model selection procedure, our model seems to report significant shifts only, that are more easily interpretable than the solutions found by other methods, and that do not rely on any chosen prior.

\begin{table}[!ht]
\begin{center}
\begin{tabular}{rlllll}
  \hline
 & Habitat & EM & bayou & SURFACE & OUshifts \\ 
  \hline
Number of shifts & 16 & 5 & 17 & 33 & 8 \\ 
  Number of regimes & 4 & 6 & 18 & 13 & 9 \\ 
  lnL & -133.86 & -97.59 & -91.54 & 30.38 & -79.79 \\ 
  Marginal lnL & NA & NA & -149.09 & NA & NA \\ 
  $\alpha$ ($\times h$, per my) & 9.32 & 12.76 & 36.54 & 1.72 & 3.25 \\ 
  $\ln 2 / \alpha$ (my) & 15.56 & 11.36 & 3.97 & 84.28 & 44.64 \\ 
  $\sigma^2$ ($\times h$, per my) & 6.21 & 5.57 & 11.91 & 0.72 & 2.29 \\ 
  $\gamma^2$ & 0.33 & 0.22 & 0.16 & 0.21 & 0.35 \\ 
  CPU time (min) & 65.25 & 134.49 & 136.81 & 634.16 & 8.28 \\ 
   \hline
\end{tabular}

\caption{Summary of the results obtained with several methods for the Chelonian Dataset. For bayou, the median of the posterior distributions is given.}\label{table:chelonia}
\end{center}
\end{table}
\section*{Acknowledgments}
We would like to thank C\'ecile An\'e for helpful discussions on an early draft of this manuscript.
We are grateful to the INRA MIGALE bioinformatics platform (\url{http://migale.jouy.inra.fr}) for providing the computational resources needed for the experiments. \PB{We also thank the two anonymous reviewers whose careful and critical reading greatly helped improve this manuscript.}




{\small
\bibliographystyle{plainnat}
\setlength{\bibsep}{0.1pt plus 0.3ex}
\bibliography{../../../../Biblio/BibTex/Bibliographie}
}


\end{document}

%% file: Figures/tree_pattern_hidden.pdf_tex
\begingroup%
  \makeatletter%
  \providecommand\color[2][]{%
    \errmessage{(Inkscape) Color is used for the text in Inkscape, but the package 'color.sty' is not loaded}%
    \renewcommand\color[2][]{}%
  }%
  \providecommand\transparent[1]{%
    \errmessage{(Inkscape) Transparency is used (non-zero) for the text in Inkscape, but the package 'transparent.sty' is not loaded}%
    \renewcommand\transparent[1]{}%
  }%
  \providecommand\rotatebox[2]{#2}%
  \ifx\svgwidth\undefined%
    \setlength{\unitlength}{304.24129522bp}%
    \ifx\svgscale\undefined%
      \relax%
    \else%
      \setlength{\unitlength}{\unitlength * \real{\svgscale}}%
    \fi%
  \else%
    \setlength{\unitlength}{\svgwidth}%
  \fi%
  \global\let\svgwidth\undefined%
  \global\let\svgscale\undefined%
  \makeatother%
  \begin{picture}(1,0.3997478)%
    \put(0,0){\includegraphics[width=\unitlength]{tree_pattern_hidden.pdf}}%
    \put(0.79893445,0.06599699){\makebox(0,0)[lb]{\smash{$Y_5$}}}%
    \put(0.79893445,0.14175642){\makebox(0,0)[lb]{\smash{$Y_4$}}}%
    \put(0.79893445,0.21754082){\makebox(0,0)[lb]{\smash{$Y_3=X_7$}}}%
    \put(0.79893445,0.29330025){\makebox(0,0)[lb]{\smash{$Y_2$}}}%
    \put(0.79893445,0.36908482){\makebox(0,0)[lb]{\smash{$Y_1=X_5$}}}%
    \put(0.01030956,0.24091675){\color[rgb]{0,0,0}\makebox(0,0)[lb]{\smash{$Z_1 = X_1$}}}%
    \put(0.60731528,0.33111659){\color[rgb]{0,0,0}\makebox(0,0)[lb]{\smash{$Z_4$}}}%
    \put(0.34479353,0.18866636){\color[rgb]{0,0,0}\makebox(0,0)[lb]{\smash{$Z_2$}}}%
    \put(0.44603048,0.07993063){\color[rgb]{0,0,0}\makebox(0,0)[lb]{\smash{$Z_3$}}}%
    \put(0.16562877,0.13388328){\color[rgb]{0,0,0}\makebox(0,0)[lb]{\smash{$\ell_2$}}}%
    \put(0.16226789,0.00939163){\color[rgb]{0,0,0}\makebox(0,0)[lb]{\smash{$t_3 = t_{8,9}$}}}%
    \put(0.65416206,0.1062313){\color[rgb]{0,0,0}\makebox(0,0)[lb]{\smash{$d_{8,9}$}}}%
    \put(0.36300131,0.26488691){\color[rgb]{0,0,0}\makebox(0,0)[lb]{\smash{$\tau_k = b_{7}$}}}%
    \put(0.41640525,0.16922374){\color[rgb]{0,0,0}\makebox(0,0)[lb]{\smash{$2=\pa(3)=\pa(7)$}}}%
    \put(0.27894575,0.35768955){\color[rgb]{0,0,0}\makebox(0,0)[lb]{\smash{$\ell_4$}}}%
  \end{picture}%
\endgroup%

%% file: Figures/basic_equivalencies_grey_small_BM.pdf_tex
\begingroup%
  \makeatletter%
  \providecommand\color[2][]{%
    \errmessage{(Inkscape) Color is used for the text in Inkscape, but the package 'color.sty' is not loaded}%
    \renewcommand\color[2][]{}%
  }%
  \providecommand\transparent[1]{%
    \errmessage{(Inkscape) Transparency is used (non-zero) for the text in Inkscape, but the package 'transparent.sty' is not loaded}%
    \renewcommand\transparent[1]{}%
  }%
  \providecommand\rotatebox[2]{#2}%
  \ifx\svgwidth\undefined%
    \setlength{\unitlength}{450.79260404bp}%
    \ifx\svgscale\undefined%
      \relax%
    \else%
      \setlength{\unitlength}{\unitlength * \real{\svgscale}}%
    \fi%
  \else%
    \setlength{\unitlength}{\svgwidth}%
  \fi%
  \global\let\svgwidth\undefined%
  \global\let\svgscale\undefined%
  \makeatother%
  \begin{picture}(1,0.20467362)%
    \put(0,0){\includegraphics[width=\unitlength]{basic_equivalencies_grey_small_BM.pdf}}%
    \put(0.27819934,0.18405118){\color[rgb]{0,0,0}\makebox(0,0)[lb]{\smash{$\mu$}}}%
    \put(0.22540356,0.06921726){\color[rgb]{0,0,0}\makebox(0,0)[lb]{\smash{$\delta_1$}}}%
    \put(0.32525436,0.06921726){\color[rgb]{0,0,0}\makebox(0,0)[lb]{\smash{$\delta_2$}}}%
    \put(0.30202566,0.00640144){\color[rgb]{0,0,0}\makebox(0,0)[lb]{\smash{$\mu\!+\!\delta_2$}}}%
    \put(0.21475571,0.00640144){\color[rgb]{0,0,0}\makebox(0,0)[lb]{\smash{$\mu\!+\!\delta_1$}}}%
    \put(0.02394249,0.18405118){\color[rgb]{0,0,0}\makebox(0,0)[lb]{\smash{$\mu$}}}%
    \put(0.07147698,0.06921724){\color[rgb]{0,0,0}\makebox(0,0)[lb]{\smash{$\delta_2 - \delta_1$}}}%
    \put(0.04971119,0.12322123){\color[rgb]{0,0,0}\makebox(0,0)[lb]{\smash{$\delta_1$}}}%
    \put(0.55985923,0.18405118){\color[rgb]{0,0,0}\makebox(0,0)[lb]{\smash{$\mu + \delta_1$}}}%
    \put(0.76198574,0.06814586){\color[rgb]{0,0,0}\makebox(0,0)[lb]{\smash{$\delta_1-\delta_2$}}}%
    \put(0.8501048,0.18405118){\color[rgb]{0,0,0}\makebox(0,0)[lb]{\smash{$\mu + \delta_2$}}}%
    \put(0.61081521,0.00640144){\color[rgb]{0,0,0}\makebox(0,0)[lb]{\smash{$\mu\!+\!\delta_2$}}}%
    \put(0.52354526,0.00640144){\color[rgb]{0,0,0}\makebox(0,0)[lb]{\smash{$\mu\!+\!\delta_1$}}}%
    \put(0.63624954,0.06921727){\color[rgb]{0,0,0}\makebox(0,0)[lb]{\smash{$\delta_2 - \delta_1$}}}%
  \end{picture}%
\endgroup%